\documentclass[twoside,leqno,twocolumn]{article}

\usepackage[letterpaper]{geometry}

\usepackage{ltexpprt}
\usepackage{hyperref}

\usepackage{algorithm}
\usepackage{algorithmic}
\usepackage{graphicx}
\usepackage{amsmath}
\usepackage{amsfonts}
\usepackage{subfigure}
\usepackage{xcolor}
\usepackage{booktabs} 

\newcommand{\PATHATTACK}{\texttt{PATHATTACK}}
\DeclareMathOperator*{\argmin}{arg\,min}
\newcommand{\ints}{\mathbb{Z}}
\newcommand{\reals}{\mathbb{R}}
\newcommand{\zeros}{\mathbf{0}}
\newcommand{\calD}{\mathcal{D}}
\newcommand{\calB}{\mathcal{B}}
\newcommand{\calP}{\mathcal{P}}
\newcommand{\calG}{\mathcal{G}}
\newcommand{\calX}{\mathcal{X}}
\newcommand{\expect}{\mathbb{E}}
\newcommand{\cerr}{c_\mathrm{err}}
\newcommand{\cVec}{\mathbf{c}}
\newcommand{\xVec}{\mathbf{x}}
\newcommand{\wVec}{\mathbf{w}}
\newcommand{\Etemp}{E_{\textrm{temp}}}
\newcommand{\btemp}{b_\mathrm{temp}}
\newcommand{\pattack}{p_\mathrm{attack}}
\newcommand{\ppath}{p_\mathrm{path}}
\newcommand{\ppair}{p_\mathrm{pair}}

 \newproof{@sketch}{Proof Sketch}
 \newenvironment{sketch}{\begin{@sketch}}{\end{@sketch}}

\begin{document}

\newcommand\relatedversion{}
\renewcommand\relatedversion{\thanks{The full version of the paper---including all appendices---can be accessed at \protect\url{https://arxiv.org/abs/2305.19083}}} 

\title{\Large Defense Against Shortest Path Attacks}
\author{Benjamin A. Miller\thanks{Northeastern University, Boston, MA, USA (\{miller.be, shafi.z, t.eliassirad\}@northeastern.edu)}\and Zohair Shafi$^*$\and Wheeler Ruml\thanks{University of New Hampshire, Durham, NH, USA (ruml@cs.unh.edu)}\and Yevgeniy Vorobeychik\thanks{Washington University in St. Louis, St. Louis, MO, USA (yvorobeychik@wustl.edu)}\and Tina Eliassi-Rad$^*$\and Scott Alfeld\thanks{Amherst College, Amherst, MA, USA (salfeld@amherst.edu)}}

\date{}

\maketitle




\fancyfoot[R]{\scriptsize{Copyright \textcopyright\ 2025\\
Copyright for this paper is retained by authors}}



\begin{abstract} \small\baselineskip=9pt Identifying shortest paths between nodes in a network is an important task in many applications. Recent work has shown that a malicious actor can manipulate a graph to make traffic between two nodes of interest follow their target path. In this paper, we develop a defense against such attacks by modifying the edge weights that users observe. The defender must balance inhibiting the attacker against any negative effects on benign users. Specifically, the defender’s goals are: (a) recommend the shortest paths to users, (b) make the lengths of the shortest paths in the published graph close to those of the same paths in the true graph, and (c) minimize the probability of an attack. We formulate the defense as a Stackelberg game in which the defender is the leader and the attacker is the follower.  We also consider a zero-sum version of the game in which the defender’s goal is to minimize cost while achieving the minimum possible attack probability. We show that the defense problem is NP-hard and propose heuristic solutions for both the zero-sum and non-zero-sum settings. By relaxing some constraints of the original problem, we formulate a linear program for local optimization around a feasible point. We present defense results with both synthetic and real networks and show that our methods often reach the lower bound of the defender’s cost.
\end{abstract}

\section{Introduction}
In numerous applications involving the routing of resources through a network, finding the shortest path between two nodes is an important problem. A malicious actor with the capacity to modify the graph could entice users to follow a particular path that could put them at risk. In cybersecurity, for example, an attacker could convince users to use compromised routers to intercept traffic and possibly steal resources~\cite{Goodin2022}. To counter adversarial activity, it is important to consider defensive measures against such behavior.
 
Recent work has proposed an algorithm to manipulate the shortest path when the attacker is able to remove edges~\cite{Miller2021}. In this paper, taking inspiration from  differential privacy, we propose a defense technique based on perturbing edge weights. Users are presented an altered set of edge weights that aims to provide the shortest paths possible while raising the attacker's cost. The contributions of this paper are as follows: (1) We define a defender cost based on the impact on user experience and probability of attack. (2) We formulate a Stackelberg game to optimize the defender's expected cost. (3) In a zero-sum setting, we show that this optimization is NP-hard. (4) We propose \texttt{PATHDEFENSE}, a heuristic algorithm that greedily increments edge weights until the user's cost is sufficiently low. (5) We present results on simulated and real networks demonstrating the cost improvement \texttt{PATHDEFENSE} provides.

\section{Problem Definition}
In our problem setting, a graph $G$ has weights $w$, and an attacker intends to remove edges to make a particular target path be the shortest between its endpoints. The defender's goal is to publish approximate weights that provide users with short paths to their destinations while also increasing the burden on the adversary, making an attack less likely. This method is inspired by a differential privacy technique for approximating shortest paths without revealing true weights~\cite{Sealfon2016}, though here we consider the weight perturbations in an optimization context. A simple example of this scenario is shown in Figure~\ref{fig:example}, which demonstrates that the defender can raise the attacker's required budget and the risk that there may be more disruption to the graph if the attack still occurs. We refer to the problem of minimizing the defender's cost in this context as the \emph{Cut Defense} problem. The analysis over the remainder of the paper makes the following assumptions: (1) The attacker has a single target path $p^*$ (not necessarily known to the defender) and uses a method known to the defender to optimize the attack. (2) If the optimization method identifies an attack within the attacker's budget $b$ (not necessarily known to the defender), the attack will occur. (3) True edge weights and removal costs are known to the attacker.
\begin{figure}
    \centering
    \includegraphics[width=0.45\textwidth]{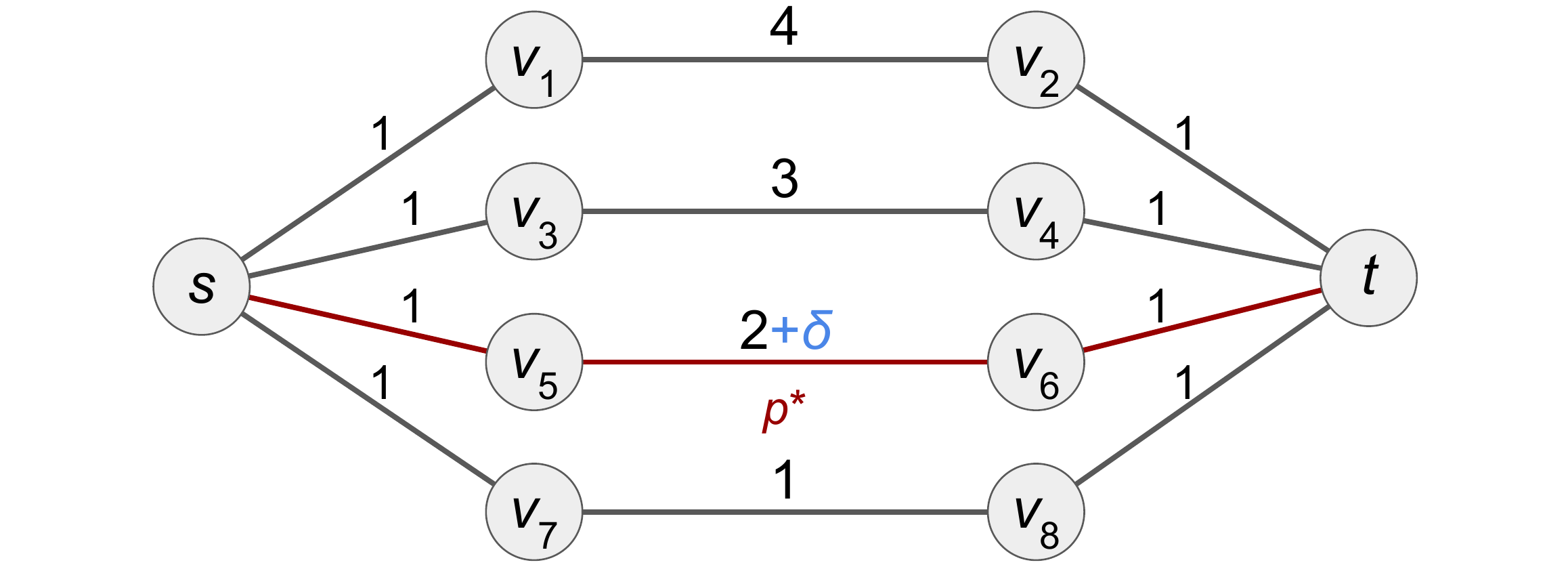}
    \caption{A simple example of the defense method. With no perturbation ($\delta=0$), $p^*$ can become the shortest path from $s$ to $t$ if only 1 edge is cut, whereas for $\delta\geq 2$, 3 edges must be removed. If the attacker has a budget of at least 3, however, the attack would cause more disruption, and the resulting cost to the defender would be higher. For example, if $\{s,v_1\}$, $\{s, v_3\}$, and $\{s,v_7\}$ are cut, all traffic between $s$ and $v_1$ or  $v_3$ will take a much longer path than it would have when $\delta=0$.}
    \label{fig:example}
\end{figure}

\subsection{Notation}
We consider a graph $G=(V, E)$, which may be directed or undirected. Each edge has a nonnegative weight $w:E\rightarrow\reals_{\geq0}$. The weights denote true traversal distances. The defender publishes weights $w^\prime:E\rightarrow\reals_{\geq0}$, which may be different than $w$. For a given source--destination pair $s,t\in V$, let $p(G, \hat{w}, s, t)$ be the shortest path in $G$ from $s$ to $t$ using weights $\hat{w}$. For a path $p$ between two nodes, let $\ell(G, \hat{w}, p)$ be the length of $p$ in $G$ using weights $\hat{w}$. We denote by $p^*$ and $b$ the attacker's target path and budget, respectively.

When determining the impact on users, we consider the distribution of source--destination pairs, $\mathcal{D}$, as this will help determine how often paths are disrupted. In addition, we assume the defender has uncertainty about $p^*$ and $b$. The defender considers a distribution $\mathcal{P}$ of possible target paths and a distribution $\mathcal{B}$ of possible budgets. These distributions result in a distribution of user-observed graphs, $\mathcal{G}$, which we describe in the next section. The defender's cost (loss) function is denoted by $L$. A notation table is provided in Appendix~\ref{sec:notation}.

\subsection{Stackelberg Game}
\label{subsec:stackelberg}
We frame the attacker--defender interaction as a Stackelberg game in which the defender is the leader and the attacker is the follower. The defender has full knowledge of the attacker's action set, and tries to choose the optimal defense given the attacker's assumed response. 

\paragraph{Attacker} The attacker will observe a graph $G=(V,E)$ with weights $w^\prime$ published by the defender, and may also know the true weights $w$. Each edge $e\in E$ has a removal cost $c(e) > 0$ that is known to the attacker. The attacker has a target path $p^*$ from source $s$ to destination $t$, and a budget $b$ specifying the greatest edge removal cost the attacker can incur. The attacker runs an algorithm, which is known to the defender, that solves the Force Path Cut problem~\cite{Miller2021}: Find a set of edges $E^\prime$ where $c(E^\prime):=\sum_{e\in E^\prime}{c(e)}\leq b$ and $p^*$ is the shortest path from $s$ to $t$ in $G^\prime=(V, E\setminus E^\prime)$ (using the published weights $w^\prime$). If the attack algorithm yields a solution with cost greater than $b$, the attack is not worth the cost to the attacker, so $G^\prime=G$.

\paragraph{Defender} The defender publishes a modified set of weights $w^\prime$. While the defender knows the method that the attacker will use, we assume there is uncertainty with respect to the attacker's target path $p^*$ and budget $b$. The defender has a distribution over both of these variables, as defined above. The distributions $\calP$ and $\calB$ combine with the published weights $w^\prime$ to create a distribution over graphs $\calG$ as follows. For a given $p^*$ in $\calP$, let $E^\prime$ be the solution given by the attack algorithm using the published weights, and assume it is a unique solution across all target paths. (If multiple target paths have the same solution, the probability of the resulting graph integrates across those paths.) Then the probability that users observe graph $G^\prime=(V, E\setminus E^\prime)$ is
\begin{equation}
    \Pr_\calG[G^\prime]=\Pr_{P\sim\calP}[P=p^*]\cdot\Pr_{B\sim\calB}[c(E^\prime)\leq B].
\end{equation}

The defender's goal is to publish a set of weights that has minimal expected cost, i.e., 
\begin{equation}
    \hat{w}^\prime=\argmin_{w^\prime}\expect\left[L(G, w, w^\prime, \calD, \calP, \calB)\right].
\end{equation}
There are several considerations when defining the defender's cost, which we discuss in detail next.

\subsection{Defender's Cost Function}
The attacker's cost function is simple: The goal is to execute the attack, so after computing the set of edges to remove, if it is within the attacker's budget, the cost is 0 (attack occurs), and otherwise the cost is 1 (attack does not occur). When determining the best course of action, the defender has three considerations. The first is the cost incurred by users of the network: the distance they must travel to get from their origin points to their destinations. If the users must travel longer distances, the cost to the defender is higher. Note that this is the actual distance traveled: The user selects a path $p$ based on the perturbed weights $w^\prime$, but the distance is computed based the original weights $w$. There is also a cost associated with users traveling a different distance than advertised. If the length of $p$ is $\ell^{\textrm{true}}$, but the user is told the length is $\ell^{\textrm{obs}}$, this may negatively affect the user's experience. If $\ell^{\textrm{obs}} < \ell^{\textrm{true}}$, the user will likely be dissatisfied with traversing a longer distance than advertised. The case where $\ell^{\textrm{obs}} > \ell^{\textrm{true}}$ is less clear. If the advertised distance is only slightly greater than the true distance, the user may be happy to experience a shorter distance than advertised. If, on the other hand, the advertised distance is drastically larger, this may induce an additional burden on users, and thus additional cost for the defender.

Finally, there may be situations where there is some additional cost to the defender if the adversary is successful. This would be a cost \emph{in addition} to the cost due to longer distances experienced by users after the attack. If, for example, the new traffic route allows the adversary to gain a competitive advantage over the defender, this would have a broader negative consequence than the specific issue of users experiencing longer distances. If this is an issue for the defender, there will be another component to the cost function to account for the expected cost of attacker success.

To formalize the cost function, we consider the three costs described above:
\begin{enumerate}
    \item $L_d$: The average \emph{distance} traveled by users\label{item:distance}
    \item $L_e$: The average cost of the \emph{error} between advertised and true path distances\label{item:error}
    \item $L_s$: The expected cost of attacker \emph{success}\label{item:success}
\end{enumerate}
Cost \ref{item:distance} takes the expected value across source--destination pairs $u,v\sim\mathcal{D}$. While the path $p$ from $u$ to $v$ is determined using the observed weights $w^\prime$, the distance experienced by users is based on the true weights $w$. Thus, for a user traveling from $u$ to $v$, we use the path $p(G^\prime, w^\prime, u, v)$, which has length $\ell(G^\prime, w, p(G^\prime, w^\prime, u, v))$. Aggregating across all pairs, cost \ref{item:distance} is expressed as 
{\small\begin{equation}
    L_d(G, w, w^\prime, \mathcal{B}, \mathcal{D}, \mathcal{P})=\expect\left[\ell(G^\prime, w, p(G^\prime, w^\prime, u, v))\right],
\end{equation}}%
where the expectation is taken over 
\begin{equation}
   s,t\sim\mathcal{D}\textrm{ and }G^\prime\sim\mathcal{G}(G, w^\prime, \mathcal{B}, \mathcal{P}).\label{eq:distributions}
\end{equation}

Cost \ref{item:error} considers the same path as cost \ref{item:distance}, but rather than the distance traveled, the defender considers a function $\cerr$ of the error between the advertised and true path lengths. Then cost \ref{item:error} is given by 
\begin{equation}
    L_e(G, w, w^\prime, \mathcal{B}, \mathcal{D}, \mathcal{P})=\expect\left[\cerr(\ell^{\textrm{true}}, \ell^{\textrm{obs}})\right],\label{eq:errorLoss}
\end{equation}
where $\ell^{\textrm{true}}=\ell(G^\prime, w, p(G^\prime, w^\prime, u,v))$ and $\ell^{\textrm{obs}}=\ell(G^\prime, w^\prime, p(G^\prime, w^\prime, u,v))$, and the expectation is once again taken over (\ref{eq:distributions}).
The shape of $\cerr$ will vary based on the defender's belief about users' degree of dissatisfaction with errors in reported path lengths. Here we use the following $\cerr$ function, where $f_+,f_->0$ denote different marginal costs for over- or under-stating the path length, respectively:
\setlength{\abovedisplayskip}{6pt}
\setlength{\belowdisplayskip}{6pt}
\begin{align}
    \cerr(\ell^{\textrm{true}}, \ell^{\textrm{obs}})=\begin{cases}
    f_+(\ell^{\textrm{obs}}-\ell^{\textrm{true}}) &\textrm{ if } \ell^{\textrm{obs}}\geq\ell^{\textrm{true}}\\
    f_-(\ell^{\textrm{true}}-\ell^{\textrm{obs}}) &\textrm{ if } \ell^{\textrm{obs}}<\ell^{\textrm{true}}
    \end{cases}
\end{align}

Finally, cost~\ref{item:success} occurs if the attack is successful. The defender has a parameter $\lambda\geq 0$ that denotes the cost of attacker success. The cost to the defender is as follows, where $p^*=p(G^\prime, w^\prime, s_{p^*}, t_{p^*})$ only if $p^*$ is the \emph{unique} shortest path from $s_{p^*}$ to $t_{p^*}$:
\begin{align}
    L_s=\lambda\Pr[p^*=p(G^\prime, w^\prime, s_{p^*}, t_{p^*})].
\end{align}

If the only cost of an attack is the direct disruption to users accounted for in $L_d$ and $L_e$, then the defender sets $\lambda=0$. A pseudocode description of an algorithm for computing the cost is in Appendix~\ref{sec:defenderCost}.

\section{Optimization}
\label{sec:optimization}
We begin by formulating the optimization to solve Cut Defense. We then define a zero-sum version in which the defender's goal is to reduce cost given that the probability of attack is minimized. We propose a heuristic method that results in a feasible solution for a single target path, then extend its usage to multiple target paths. We finally derive a linear program for local optimization around a feasible point.

\subsection{Non-Convex Optimization Formulation}
\label{subsec:nonconvex}
We optimize cost while varying perturbed weights. Let $\wVec\in\reals_{\geq0}^{|E|}$ be the vector of original edge weights, where each edge is given an arbitrary index corresponding to its vector entry. The vector $\wVec^\prime$ contains the perturbed weights, $\cVec$ contains edge removal costs, and $\xVec_p$ is a binary indicator vector for path $p$, i.e., if the $i$th edge is in path $p$, the $i$th entry in $\xVec_p$ is 1, otherwise it is 0. Let $P(u, v)$ be the set of all paths from $u$ to $v$, $q_{uv}=\Pr_{D\sim\calD}[D=(u,v)]$ be the probability that a randomly selected user travels from $u$ to $v$, $P_t=\mathrm{support}(\calP)$ be all paths with nonzero probability of being the target, $W$ be a large value to denote edge removal, and $\calX(G, \wVec, p^*)$ be the set of attacks against graph $G$ with weights $\wVec$ to make $p^*$ be the shortest path between its terminal nodes. We solve Cut Defense by optimizing as follows:
{\small\begin{align}
    \hat{\wVec}^\prime=&\argmin_{\wVec^\prime}{\lambda\left(1-z_\emptyset\right)+\sum_{u,v\in V}L_d(u, v)+L_e(u,v)}\\
    \textrm{s.t.}&~L_d(u,v)=q_{uv}\cdot\sum_{p^*\in{P_t\cup\{\emptyset\}}}{z_{p^*}\cdot \ell_{uv,p^*}^\mathrm{true}}\label{eq:distance_cost}\\
    &~~~~~~~~~~~~~~~~~~~~~~~~\forall u,v\in V\nonumber\\
    &~L_e(u,v)=q_{uv}\cdot\sum_{p^*\in{P_t\cup\{\emptyset\}}}{z_{p^*}\cdot L_e(u,v,p^*)}\label{eq:error_cost_total}\\
    &~~~~~~~~~~~~~~~~~~~~~~~~\forall u,v\in V\nonumber
\end{align}}%
{\small\begin{align}
    &~L_e(u,v,p^*)=\left(f_+\cdot d_{uv,p^*}^\mathrm{pos} + f_-\cdot d_{uv,p^*}^\mathrm{neg}\right)\label{eq:error_cost_single}\\
    &~~~~~~~~~~~~~~~~~~~~~~~~\forall u,v\in V, p^*\in P_t\cup\{\emptyset\}\nonumber\\
    &~\Delta_{p^*}\in\calX(G, \wVec^\prime, p^*)\ \ \ \forall p^*\in P_t\label{eq:first_attack_const}\\
    &~\Delta_{\emptyset}=\zeros\label{eq:no_attack}\\
    &~\cVec^\top\Delta_{p^*}\leq \cVec^\top\Delta\ \ \ \forall \Delta\in\calX(G, \wVec^\prime, p^*), p^*\in P_t\label{eq:minDelta}\\
    &~z_{p^*}=\Pr(p^*)\sum_{i\geq\cVec^\top\Delta_{p^*}}{\Pr_{B\sim\calB}[B=i]}\ \ \ \forall p^*\in P_t\label{eq:attack_prob}\\
    &~z_\emptyset=1-\sum_{p^*\in P_t}z_{p^*}\label{eq:z0}\\
    &~z_{p^*}\geq 0\ \ \ \forall p^*\in P_t\cup\{\emptyset\}\label{eq:last_attack_const}\\
    &~p_{uv,p^*}=\argmin_{p\in P(u,v)}\xVec_p^\top(\wVec^\prime+W\Delta_{p^*})\label{eq:minpath}\\
    &~~~~~~~~~~~~~~~~~~~~~~~~\forall u,v\in V,p^*\in P_t\cup\{\emptyset\}\nonumber\\
    &~\ell_{uv,p^*}^\mathrm{true}=\xVec_{p_{uv,p^*}}^\top\wVec\ \ \ \forall u,v\in V,p^*\in P_t\cup\{\emptyset\}\label{eq:true_length}\\
    &~\ell_{uv,p^*}^\mathrm{obs}=\xVec_{p_{uv,p^*}}^\top\wVec^\prime\ \ \ \forall u,v\in V,p^*\in P_t\cup\{\emptyset\}\label{eq:obs_length}\\
    &~d_{uv,p^*}^\mathrm{pos},d_{uv,p^*}^\mathrm{neg}\geq 0\ \ \ \forall u,v\in V,p^*\in P_t\cup\{\emptyset\}\label{eq:diff_piece}\\
    &~d_{uv,p^*}^\mathrm{pos}-d_{uv,p^*}^\mathrm{neg}=\ell_{uv,p^*}^\mathrm{obs}-\ell_{uv,p^*}^\mathrm{true}\label{eq:diff_diff}\\
    &~~~~~~~~~~~~~~~~~~~~~~~~\forall u,v\in V,p^*\in P_t\cup\{\emptyset\}.\nonumber
\end{align}}%
Note that $p_{uv,p^*}$, defined in (\ref{eq:minpath}), is the shortest path from $u$ to $v$ according to the published weights after the attacker attacks when the target path is $p^*$. An explanation of each constraint is provided in Appendix~\ref{sec:NP_constraints}. The case where there is no attack is considered when $p^*=\emptyset$. 

One potential concern when calculating the expected cost across pairs of nodes is the possibility that the graph could become disconnected, leaving some inter-node distances infinite. However, due to the following theorem, we do not need to be concerned about this possibility. A proof and a discussion of practical issues are provided in Appendix~\ref{sec:connectedProof}.
\begin{theorem}
The optimal $\Delta_{p^*}$ in (\ref{eq:minDelta}) will not disconnect $G$.\label{thm:connected}
\end{theorem}

\subsection{Zero-Sum Formulation}
\label{subsec:zerosum}
In the prior section, we assumed a non-zero-sum game in which the optima for the attacker and defender may coincide. We gain additional insight into the problem by considering the zero-sum version of the problem, in which the defender's primary goal is ensuring the attack does not occur. In this case, we are given the same information as in Cut Defense except the cost of attack success $\lambda$. Instead, the defender manipulates the weights $w^\prime$ to minimize the probability of attack, i.e.,
{\small
\begin{align}
    z_\textrm{min}=&\min_{w^\prime}{\sum_{p^*\in\calP}\Big(\Pr_{P\sim\calP}[P=p^*]}\\
    &~~~~~~~~~~~~~~\cdot\Pr_{B\sim\calB}[c(E^\prime(G, w^\prime, p^*))\leq B]\Big).\nonumber
\end{align}}%
Within the minimized attack probability, however, the defender wants the cost to be as low as possible. Thus, the \emph{Zero-Sum Cut Defense} problem is given by
{\small\begin{align}
    \hat{w}^\prime=&\argmin_{w^\prime}L_d(G, w, w^\prime, \mathcal{B}, \mathcal{D}, \mathcal{P})\\
    &~~~~~~~~~~~~~+L_e(G, w, w^\prime, \mathcal{B}, \mathcal{D}, \mathcal{P})\nonumber\\
    \textrm{s.t.}&~z_\textrm{min}=\sum_{p^*\in\calP}\Big(\Pr_{P\sim\calP}[P=p^*]\\
    &~~~~~~~~~~~~~~~~~~~\cdot\Pr_{B\sim\calB}[c(E^\prime(G, w^\prime, p^*))\leq B]\Big).\nonumber
\end{align}}%
Note that $L_s$ is not considered in the objective in this formulation, since the attack probability is fixed at its minimum possible value. We show that this version of the problem is NP-hard.
\begin{theorem}
Zero-Sum Cut Defense is NP-hard.
\label{thm:hardness}
\end{theorem}
\begin{sketch}
We prove NP hardness via reduction from the Knapsack problem. Given a set of $n$ items with values $\nu_i\in\ints_+$ and weights $\eta_i\in\ints_+$, $1\leq i\leq n$, and two thresholds $U$ and $H$, the Knapsack problem is to determine whether there is a subset of items with total value at least $U$ with weight no more than $H$.
For each item, we create a triangle in a graph, where consecutive triangles share a node as shown in Fig.~\ref{fig:reduction}. The $i$th triangle consists of the nodes $u_{i-1}$, $u_i$, and $\omega_i$. Let $s=u_0$ and $t=u_n$. We create a Zero-Sum Cut Defense instance in which the support of $\calP$ consists of the single path from $s$ to $t$ that passes through no nodes $\omega_i$ for any $1\leq i\leq n$.  For all $i$, edge $\{u_{i-1}, u_i\}$ has weight 1 and removal cost 1, $\{u_{i-1}, \omega_i\}$ has weight 1 and cost $v_i$, and edge $\{\omega_i, u_i\}$ has weight $w_i$ and cost $v_i$. The adversary's budget is $U-1$ with probability 1. The defender's cost only considers traffic going from $s$ to $t$, i.e., $\Pr_{(x,y)\sim\calD}[(x,y)=(s,t)]=1$. Let $f_+=1$ and  $f_-=H^\prime=\sum_{i=1}^n{w_i}$.
\begin{figure}
    \centering
    \includegraphics[width=0.49\textwidth]{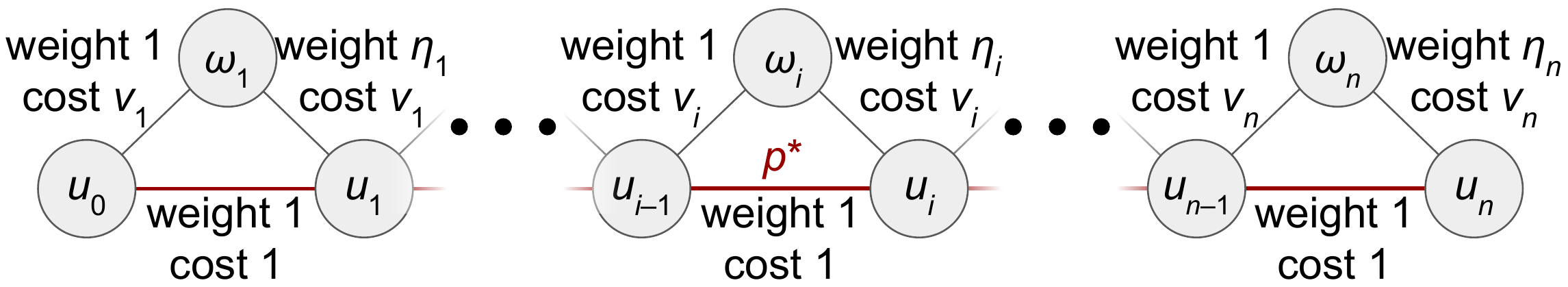}
    \caption{Reduction from Knapsack to Zero-Sum Cut Defense. The $i$th item in the set corresponds to a triangle $\{u_{i-1}, \omega_i, u_i\}$. All target paths go from $s=u_0$ to $t=u_n$. The target path $p^*$ traverses the bottom edges on the figure, highlighted in red. Keeping defender cost low while ensuring the probability of attack is 0 is equivalent to keeping the knapsack's weight low while ensuring the value is sufficient.}
    \label{fig:reduction}
\end{figure}

Since the adversary's budget is $U-1$, in order to minimize the attack probability (in this case, make it 0), the defender must force some subset of edges along $p^*$ to have length at least the same as the two-hop paths running parallel to them. The removal costs on the parallel paths of these edges must sum to at least $U$: the defender's ``value'' is an increased cost for the attacker. The increase in distance traveled by the user will be commensurate with the weights of the items associated with the perturbed edges. This provides a direct mapping between solving Knapsack and solving Zero-Sum Cut Defense on the generated graph.
 A detailed proof is provided in Appendix~\ref{sec:proof}.
\end{sketch}

Although the problem cannot be efficiently solved in general, we find feasible points fairly easily by increasing the length of $p^*$ until the cost of the attack is sufficiently high. Starting with weights $w^\prime$ initialized to the true weights $w$, the procedure is as follows:
\begin{enumerate}
    \item $E^\prime\gets$ attack$(G, w^\prime, p^*)$
    \item $p\gets$ 2nd shortest path between the terminals of $p^*$, if it exists
    \item pick an edge $e$ from $E_{p^*}\setminus E_p$ 
    \item increase $w^\prime(e)$ by $\delta=\ell(G, w^\prime, p)-\ell(G, w^\prime, p^*)$
\end{enumerate}
Here $E_p$ is the set of edges on path $p$. This procedure continues until either $p^*$ is the longest path between its terminals or $c(E^\prime)$ exceeds the largest possible attack budget. This procedure yields a feasible point given the attacker's algorithm. If we continue until $p^*$ is the longest path, we must be at a feasible point: all other paths that connect $p^*$'s endpoints need to be cut. This observation yields the following theorem.
\begin{theorem}
When $\calP$ consists of a single path $p^*$, the procedure above yields a feasible point for Zero-Sum Cut Defense. 
\end{theorem}

While having multiple possible target paths complicates the problem, we use a similar principle to reduce the probability of attack. For each target path, we increase the edge weights as described above. We then apply the attack and use the number of edges removed to calculate the attack probability. Then, starting from the original graph, we consider the target paths in order of increasing attack probability. We then increase the edge weights on each target path again, accumulating the new weights each time. This prioritizes the path at the end of the sequence, which has the highest probability of resulting in an attack.


\subsection{Heuristic Method}
From the zero-sum case, we see that increasing the weights on target paths is an effective strategy. Taking this as inspiration, we propose a heuristic algorithm that iteratively chooses an edge $e$ from some target path $p^*$ and increments its weight to add another path to $P_{p^*}$. We call this algorithm \texttt{PATHDEFENSE} (see Algorithm~\ref{alg:heuristic}). At each iteration, the algorithm considers edges on which the smallest possible weight increase will provide one target path with a new competing path. For a given target path $p^*$, these edges are identified by applying the attack and finding the second-shortest path $p$ between the source and destination of $p^*$, if such a path exists. The edges in $p^*$ that are not part of $p$ may be incremented to add $p$ as a competing path that must be cut to make $p^*$ shortest (see Algorithm~\ref{alg:increment}). 
The attack probability is evaluated after considering each possible perturbation, and whichever perturbation results in the smallest attack probability is kept. If multiple perturbations result in the same attack probability, the edge is chosen that maximizes the average length of $p^*$. This procedure continues until (1) all target paths are the longest between their terminals, or (2) a threshold is reached in terms of cost, attack probability, or number of iterations. 
\setlength{\floatsep}{3pt}        
\setlength{\textfloatsep}{3pt}    
\setlength{\intextsep}{3pt}       
\begin{algorithm}[H]
\small
\caption{\small\texttt{get\_edge\_increments}}
\label{alg:increment}
\textbf{Input}: graph $G=(V, E)$, perturbed weights $w^\prime$, $p^*$ dist. $\calP$

\textbf{Output}: weight increment set $R$

\begin{algorithmic}[1] 
\STATE $R\gets\emptyset$
\FORALL{$p^* \in P_t$}
\STATE $\Etemp\gets\textrm{attack}(G, w^\prime, p^*)$
\STATE $G^\prime\gets (V, E\setminus\Etemp$)
\IF{$p^*$ is not the only path from $s$ to $t$ in $G^\prime$}
\STATE $p\gets$ 2nd shortest path from $s$ to $t$ in $G^\prime$ using $w^\prime$
\STATE $\delta\gets\ell(G^\prime, w^\prime, p)-\ell(G^\prime, w^\prime, p^*)$
\FORALL{$e\in E_p\setminus E_{p^*}$}
\STATE $R\gets R\cup\{(e, \delta)\}$
\ENDFOR
\ENDIF
\ENDFOR
\RETURN $R$
\end{algorithmic}
\end{algorithm}

\begin{algorithm}[tb]
\small
\caption{\small\texttt{PATHDEFENSE} Heuristic Defense Algorithm}
\label{alg:heuristic}
\textbf{Input}: graph $G=(V, E)$, true weights $w$, budget dist. $\mathcal{B}$, $p^*$ dist. $\calP$, $(u,v)$ pair dist. $\mathcal{D}$, attack cost $\lambda$, cost threshold $\epsilon_c$, attack prob. threshold $\epsilon_a$, max. iterations $i_{\max}$

\textbf{Output}: perturbed weights $w^\prime$

\begin{algorithmic}[1] 
\STATE $w^\prime\gets w$
\REPEAT
\STATE $R\gets\texttt{get\_edge\_increments}(G, w^\prime, \calP)$
\STATE $(e_{\min}, \delta_{\min}, z_{\min}, \ell_{\min})\gets(\emptyset, 0, 2, 0)$
\FORALL{$(e,\delta)\in R$}
\STATE $w^\prime(e)\gets w^\prime(e)+\delta$
\STATE $z\gets0$
\FORALL{$p^*\in P_t$}
\STATE $\Etemp\gets\textrm{attack}(G, w^\prime, p^*)$
\STATE $z\gets z+\Pr_{P\sim\calP}[P=p^*]\Pr_{B\sim\calB}[B\geq c(\Etemp)]$
\IF{$z < z_{\min}$ \OR($z=z_{\min}$ \AND $\expect_{P\sim\calP}\left[\ell(G, w^\prime, P)\right]>\ell_{\min}$)}
\STATE $e_{\min}\gets e$
\STATE $\delta_{\min}\gets\delta$
\STATE $z_{\min}\gets z$
\STATE $\ell_{\min}\gets\expect_{P\sim\calP}\left[\ell(G, w^\prime, P)\right]$
\ENDIF
\ENDFOR
\STATE $w^\prime(e)\gets w^\prime(e)-\delta$
\ENDFOR
\STATE $w^\prime(e_{\min})\gets w^\prime(e_{\min})+\delta_{\min}$
\STATE $L_d, L_e, L_s\gets$cost$(G, w, w^\prime, \calB, \calP, \calD, \lambda)$
\STATE $i\gets i+1$
\UNTIL{$L_d+L_e+L_s<\epsilon_c$ \OR $z_{\min} < \epsilon_a$ \OR $|R|=0$ \OR $i\geq i_{\max}$}
\RETURN $w^\prime$

\end{algorithmic}
\end{algorithm}

\subsection{Local Optimization Around a Feasible Point}
Once a feasible point is identified, we relax the hardest constraints to formulate a linear program for local optimization. In this case, we fix the attack that occurs for each $p^*$, and ensure that the observed shortest path between each pair of nodes remains the same as the weights are varied. By fixing the attack, we are given a value for $\Delta_{p^*}$ and $z_{p^*}$, thus removing constraints (\ref{eq:first_attack_const})--(\ref{eq:last_attack_const}) from the nonconvex optimization in Section~\ref{subsec:nonconvex}. By fixing the shortest path, we are given a value for $p_{uv,p^*}$, replacing  constraint (\ref{eq:minpath}) with

\begin{align}
    \ell_{uv,p^*}^\mathrm{obs}&\leq\xVec_{p}^\top\left(\wVec^\prime+W\Delta_{p^*}\right)\\
    &~~~\forall u,v\in V,p^*\in P_t\cup\{\emptyset\}, p\in P(u,v).\nonumber
\end{align}
All remaining constraints in the nonconvex program are linear. This is not, however, sufficient to locally optimize: The attack $\Delta_{p^*}$ must be both necessary (not cut superfluous edges) and sufficient (cut all paths that compete with $p^*$). To optimize within this context, we add the constraints

\begin{align}
    \xVec_{p^*}^\top\wVec^\prime&\leq\xVec_{p}^\top(\wVec^\prime+W\Delta_{p^*})-\epsilon_{p^*}\label{eq:sufficient}\\
    &~~~\forall p^*\in P_t,p\in P(s_{p^*}, t_{p^*})\nonumber\\
    \xVec_{p^*}^\top\wVec^\prime&\geq\xVec_{p}^\top\wVec^\prime\ \ \ p^*\in P_t,p\in P_{p^*}\label{eq:necessary}
\end{align}
Here $P_{p^*}$ is the set of paths competing with $p^*$ to be the shortest path from $s_{p^*}$ to $t_{p^*}$.
To ensure sufficiency, (\ref{eq:sufficient}) constrains all paths between the terminals of a target path $p^*$ to be strictly longer than $p^*$. The additional variables $\epsilon_{p^*}$ may be measured based on the difference in lengths between $p^*$ and the second-shortest path after the attack. Constraint (\ref{eq:necessary}) ensures necessity by making all paths that competed with $p^*$ at the feasible point remain competitive. Since all constraints are linear, we use constraint generation to explicitly state only a subset of the necessary constraints~\cite{Ben-Ameur2006}.

\section{Experiments}
\label{sec:experiments}
We demonstrate the optimization procedure with 4 synthetic network generators and 4 real networks. All synthetic networks have 250 nodes and an average degree of approximately 12. We use Erd\H{o}s--R\'{e}nyi (ER) random graphs, Barab\'{a}si--Albert (BA) preferential attachment graphs, Watts--Strogatz (WS) small-world graphs, and stochastic blockmodel (SBM) graphs where nodes are separated into communities of size 200 and 50. To provide some variation in edge weights, all edges are given weights drawn from a Poisson distribution with rate parameter 20. Removal costs are set to 1 for all edges.
\begin{figure*}[ht]
    \centering
    \begin{subfigure}[]{}
    \centering
    \includegraphics[width=0.24\textwidth]{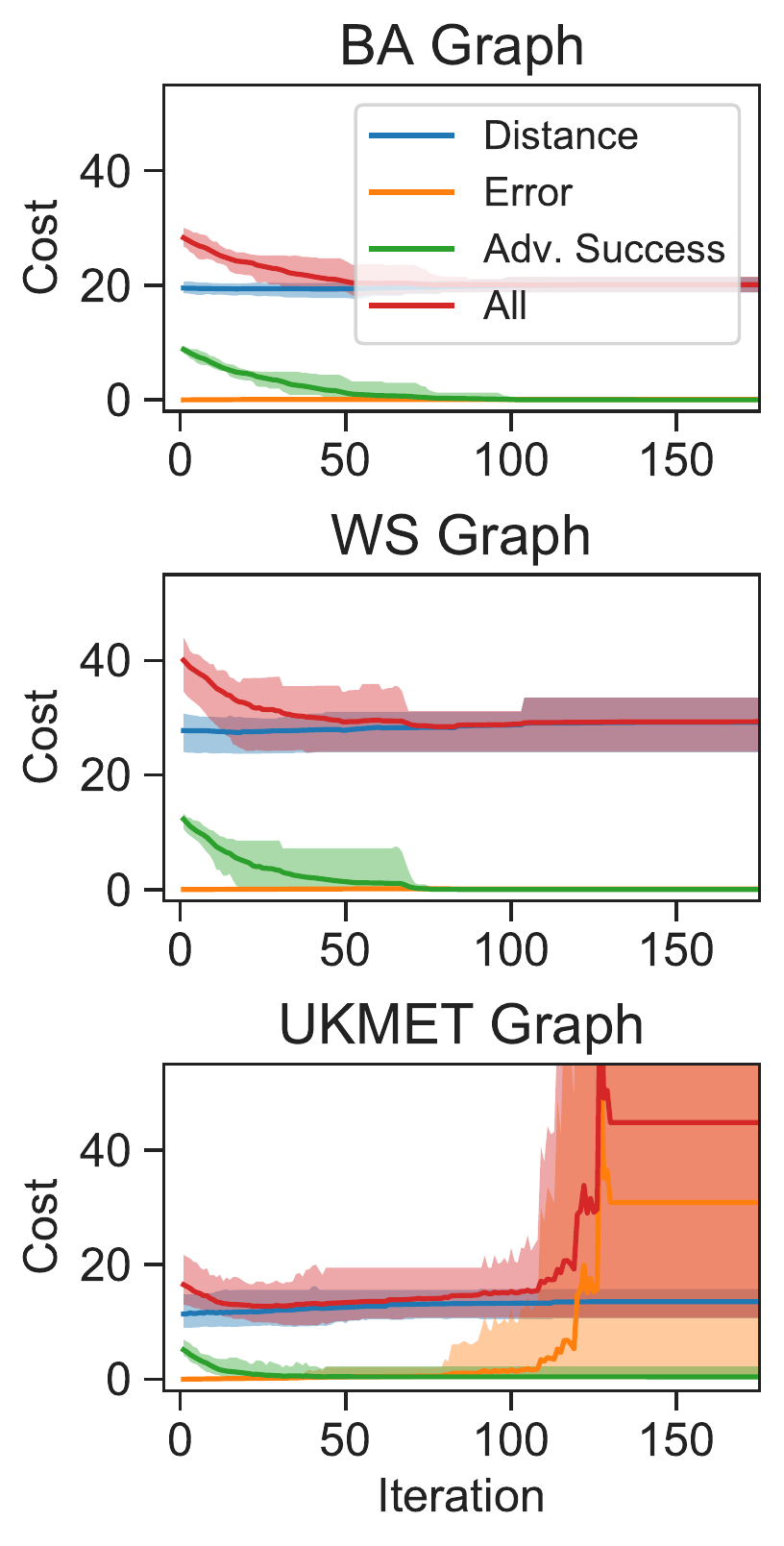} 
    \end{subfigure}
    \hfill
    \begin{subfigure}[]
    \centering
        \includegraphics[width=.7\textwidth]{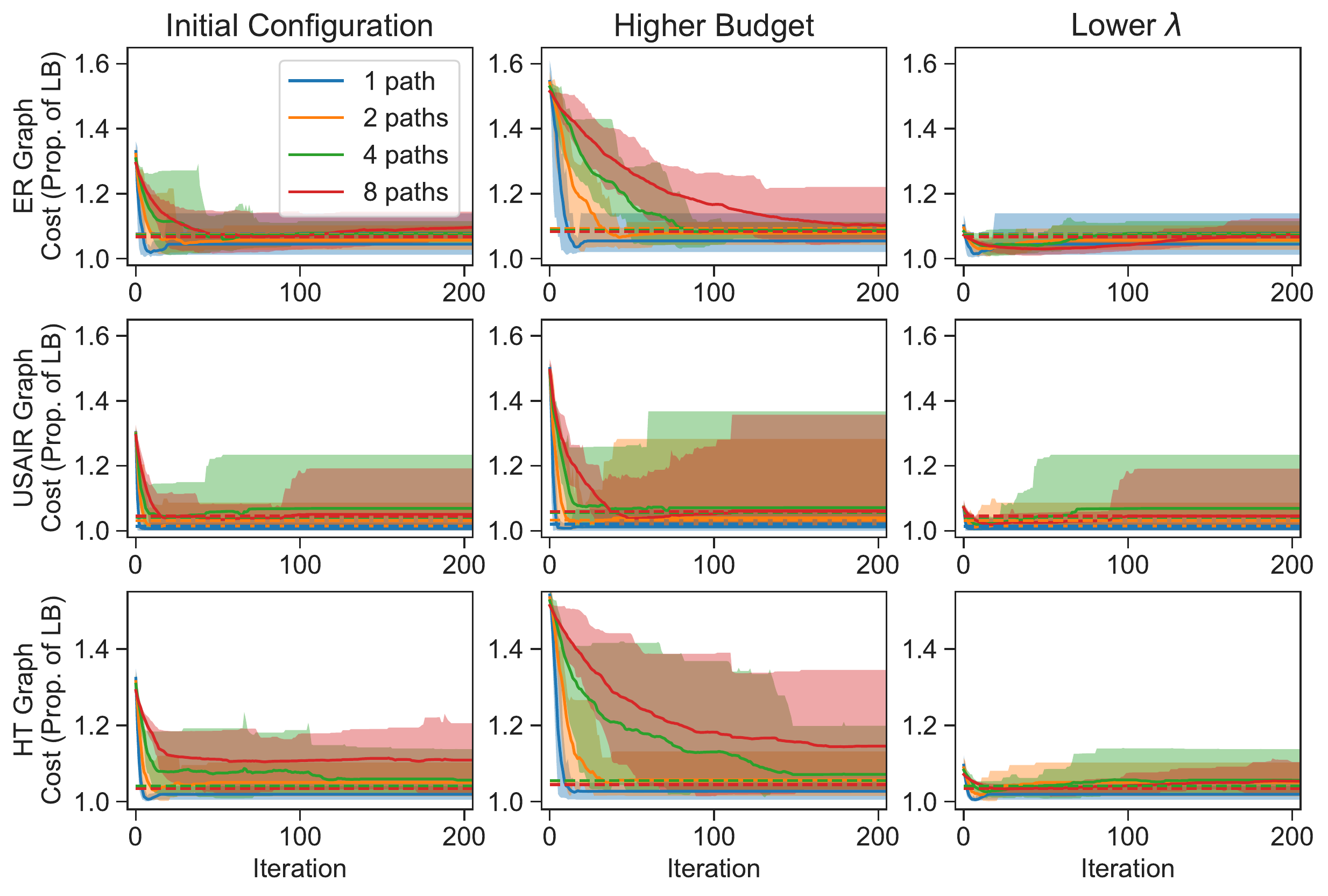}
    \end{subfigure}
    \caption{Cost of \texttt{PATHDEFENSE} when all target paths share terminals. Lower cost is better for the defender. Plots include the average cost (solid line) and the cost range across trials (shaded area). (a) Cost broken down by component for BA, WS, and UKMET graphs with 4 target paths. There is a substantial reduction in cost due to the probability of adversary success being reduced, and the cost due to errors in published distances is minimal for BA and WS, whereas $L_e$ increases a substantial amount in the UKMET data, as it is difficult to avoid traversing perturbed edges. (b) Defender costs, normalized by a lower bound, for ER (top), USAIR (middle), and HT (bottom) graphs. Results are shown for the original budget and $\lambda$ (left), when the attacker budget is doubled (center), and when $\lambda$ is reduced by five times (right). The average zero-sum result is also shown (dashed line). As expected, increasing the adversary's budget results in slower convergence, and decreasing the attack success cost reduces the improvement provided by \texttt{PATHDEFENSE}.}
    \label{fig:results}
\end{figure*}

The real network datasets include 2 transportation networks, 1 social network, and 1 computer network. The transportation networks are United States airports (USAIR), where edge weights are the number of seats on flights between airports~\cite{Colizza2007}, and United Kingdom metro stops (UKMET), where weights are travel times between stops in minutes~\cite{Gallotti2015}. The social network is attendees at the 2009 ACM Hypertext Conference (HT), where weights are the number of face-to-face interactions during the conference~\cite{Isella2011}. The computer network is an autonomous system (AS) graph~\cite{Leskovec2005}, with weights from a Poisson distribution as in the synthetic networks.
Weights for USAIR and HT are inverted to create distances rather than similarities. Statistics and links to the datasets are in Appendix~\ref{sec:datasets}. 

\subsection{Experimental Setup}
For each experiment on a given dataset, to evaluate the effects of varying the number of target paths, we choose 1, 2, 4, or 8 target paths. Source--destination pairs are chosen uniformly at random and the target paths include the 5th shortest and every second path thereafter until the desired number of targets is reached. In some cases, all target paths have the same terminal nodes; in others, we choose independently for each path. For SBM and AS graphs, we also consider the case where the two terminal nodes are from one community of nodes, but the target path traverses nodes in another one. (We call this an \emph{extra-community} path.) This emulates a scenario where an outside attacker wants the traffic to take a relatively unnatural path, e.g., computer traffic unnecessarily crossing national boundaries.

In each experiment, $\calB$ is a Poisson distribution whose rate parameter is the average number of edges removed by the attack across all target paths. For the attack, we use the version of \PATHATTACK{} that solves a relaxed linear program and randomly rounds to get an integer solution~\cite{Miller2023}. The source--destination distribution emphasizes the portions of the graph with target paths: with probability 0.5, we draw two nodes both either on a target path or on the true shortest path between its endpoints, and with probability 0.5 we do not. (Pairs are uniformly distributed within each category.) For each setting, results are aggregated across 10 trials. The cost of attacker success $\lambda$ is set to one half of the cost based on distances alone when there is no attack.
 We used a CentOS Linux cluster with 32 cores per machine for our experiments. Each job was allocated 10 GB of memory. We used Python 3.8.1 with Gurobi 9.5.1 (\url{https://www.gurobi.com}) for optimization and NetworkX 2.4 (\url{https://networkx.org}) for graph analysis.

\subsection{Results}
\label{subsec:results}
We first consider how the three components of the defender's cost vary when running \texttt{PATHDEFENSE}. Representative results are shown in Fig.~\ref{fig:results}(a). Cost is typically dominated by the true distance traveled by users. While the traffic is primarily on the portion of the graph affected by the attack, the impact of errors is negligible in comparison. One reason for this phenomenon is that increasing edge weights discourages their use: when a path looks longer, fewer users will take it and it will not be considered in the cost. In the UKMET case, however, this changes after about 100 iterations, when the cost from errors drastically increases. The metro graph is somewhat tree-like, making it difficult to avoid traversing perturbed edges. In all cases, the overall reduction in cost comes from a large reduction in the probability of attack and a small increase in the average distance traveled. 

The cost of \texttt{PATHDEFENSE} for three additional datasets is shown in Fig.~\ref{fig:results}(b). The plots include cases where the rate parameter of the budget distribution is doubled and where the cost of adversary success is reduced by a factor of five. We report all costs as a proportion of a lower bound: the cost when there is no attack and no perturbation. When the attacker's budget is doubled, the initial defender cost is much larger, but the cost eventually obtained by \texttt{PATHDEFENSE} is very similar. It is typical for \texttt{PATHDEFENSE} to outperform the zero-sum case at an early iteration when there are few target paths, though this does not always happen. When there are 8 target paths in the HT graph, the zero-sum procedure produces a better result than \texttt{PATHDEFENSE}. This may be due to the clustering that exists in the social network: it may promote oscillation between competing paths, whereas the zero-sum method focuses on one path at a time. Results on all
datasets are provided in Appendix~\ref{sec:moreResults}. Results compar-
ing with a simple baseline defense and using alternative
attack methods are provided in Appendix~\ref{sec:baseline}. 

Fig.~\ref{fig:community} shows results where the attacker targets a path that exits and re-enters a community. The lowest relative cost is higher in this case than when paths are chosen by enumerating consecutive shortest paths, which is consistent with intuition. In the SBM graph with extra-community target paths, we again see that the zero-sum method yields lower cost than \texttt{PATHDEFENSE}, suggesting that optimizing each target path in sequence is effective in this case as well.  
\begin{figure}
    \centering
    \includegraphics[width=0.45\textwidth]{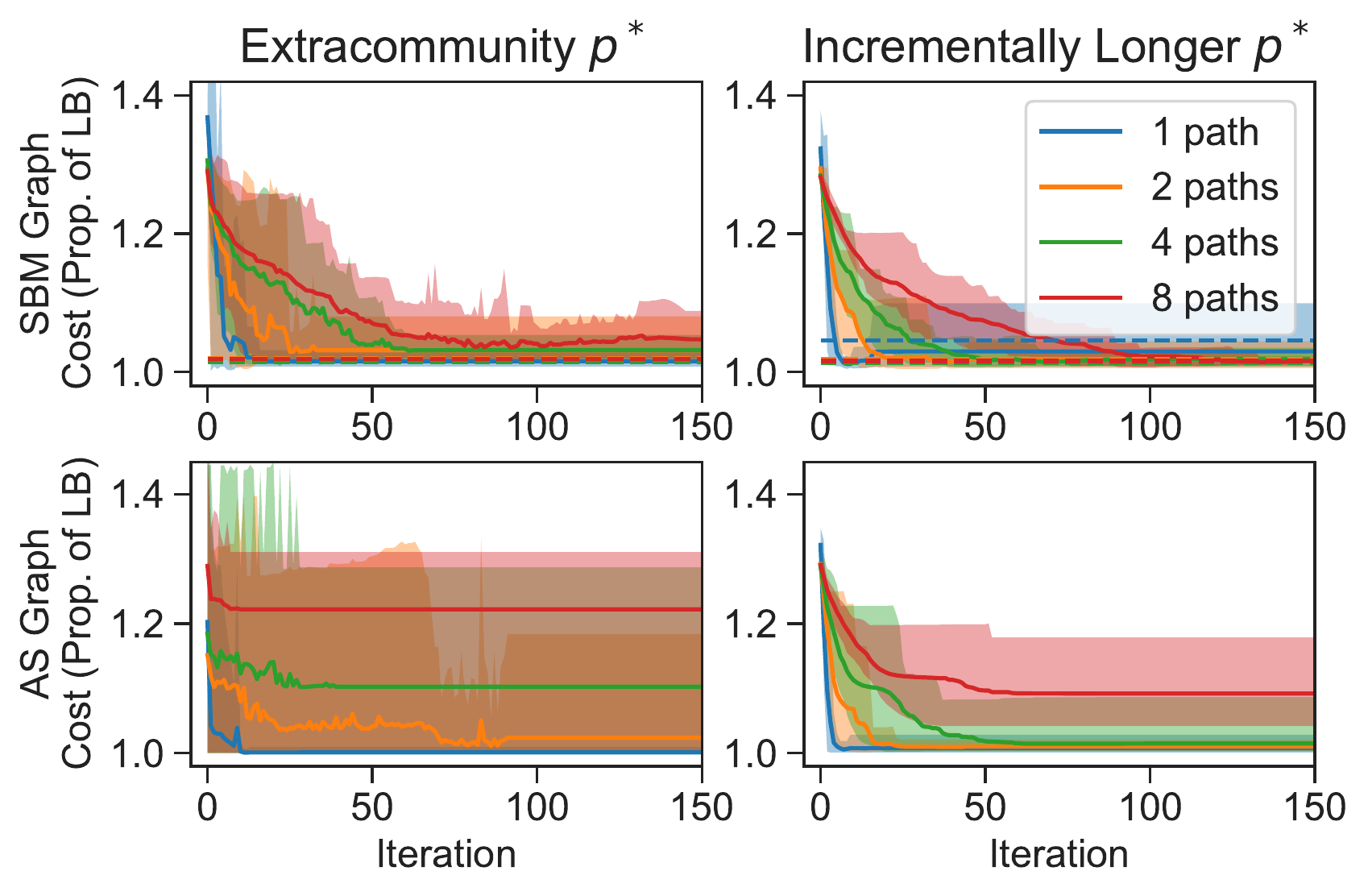}
    \caption{Results when the attacker targets an extra-community path (left), in comparison to when targets are incrementally longer short paths (right), in SBM (top) and AS (bottom) graphs. Plots include the average cost (solid line) and the cost range across trials (shaded area), as well as the average zero-sum result (dash line). (The zero-sum procedure is only reported for SBM; it did not complete within 24 hours for AS.) \texttt{PATHDEFENSE} yields lower relative cost in the case where the target path is incrementally longer than the true shortest path than in the extra-community case.}
    \label{fig:community}
\end{figure}

\section{Related Work}
\label{sec:related}
The problem of releasing a graph that can be useful while not revealing sensitive information has received much attention since the discovery of the problem of deanonymization~\cite{Backstrom2007}. Much of this research has been concerned with privacy-preserving release of social network data, where nodes are anonymized w.r.t.~topological features such as degree~\cite{Liu2008}, neighborhood~\cite{Zhou2008}, or cluster~\cite{Bhagat2009}.
 Sharing of sensitive graph data has been studied in the context of differential privacy~\cite{Dwork2006}. 
Sealfon~\cite{Sealfon2016} applied differential privacy to graph weights in the context of computing shortest paths.
Other methods have considered differential privacy for unweighted graphs. For example, a data-driven low-dimensional projection~\cite{Arora2019} and random low-dimensional projections~\cite{Blocki2012} have been applied for cut queries, i.e., calculating the number of edges that must be removed to disconnect two sets of vertices. Other  work does not necessarily preserve distances between pairs of nodes, but maintains the distribution of distances~\cite{Chen2014}.

Outside of differential privacy, there is work on reliably finding shortest paths when a graph is located on an untrusted server~\cite{Gao2011}. In other work, an actor wants to ``buy'' a path from $s$ to $t$, with the prices only known to current owners~\cite{Archer2007}. In this setting, a buyer can be forced to overpay for the path~\cite{Elkind2004}, which is similar to the Cut Defense goal of forcing an attacker to expend extra resources, though in a different data access mechanism.

The \PATHATTACK{} algorithm is an example of inverse optimization~\cite{Ahuja2001}, and specifically the inverse shortest path problem~\cite{Zhang1995}: rather than optimizing the path length for a given graph, we change the graph to make a certain path the shortest. Inverse shortest path problems have proven useful in various navigation scenarios~\cite{Brandao2021,LaFontaine2022}. Other research has considered inverse shortest path lengths~\cite{Tayyebi2016} and other inverse optimizations, such as max flow/min cut~\cite{Liu2006,Jiang2010,Deaconu2020}. 
%
 Recent work in adversarial machine learning has expanded to attacking graph structure to alter vertex classification outcomes~\cite{Zugner2018} or vertex embeddings~\cite{Bojchevski2019}. Other work has aimed to defend against such attacks by, for example, using a low-rank representation of the graph~\cite{Entezari2020} or filtering based on node attribute similarity~\cite{WuH2019}. Other adversarial graph analysis research has focused on limiting influence in networks~\cite{Medya2020} and evading centrality measures~\cite{Waniek2021}. 

\section{Conclusion and Future Work}
We present a framework and algorithms for defending against shortest path attacks. We formulate the defense as a Stackelberg game in which the defender alters the weights of the graph before the attacker removes edges to make the target path shortest. The defender's cost includes components to limit the average distance traveled by users, the error in the published distances, and the probability of attacker success. We show that the zero-sum version of this problem is NP-hard and provide a greedy edge weight increment procedure to find a feasible point. Using this same procedure, we propose the \texttt{PATHDEFENSE} algorithm and apply it to several real and synthetic datasets. Across a wide set of experiments, we observe that \texttt{PATHDEFENSE} reduces the attack probability to a negligible level (typically less than $10^{-6}$) while only slightly increasing the cost borne by users (by less than 5\% in over 87\% of cases).

There is currently no performance guarantee for \texttt{PATHDEFENSE}. Future work in this area should focus on obtaining such a guarantee, perhaps by narrowing the scope of the problem (e.g., single target path, weights only increased). Lastly, using a parallel method that exploits sparseness~\cite{Yang2023} could be key to applying \texttt{PATHDEFENSE} at scale. 
\section*{Acknowledgement}
BAM was supported by the United States Air Force under Contract No. FA8702-15-D-0001. TER was supported in part by the Combat Capabilities Development Command Army Research Laboratory (under Cooperative Agreement No.~W911NF-13-2-0045) and by the Under Secretary of Defense for Research and Engineering under Air Force Contract No.~FA8702-15-D-0001. YV was supported by grants from the Office of Naval Research (N00014-24-1-2663) and National Science Foundation (CAREER Award IIS-1905558, IIS-2214141, and CNS-2310470). Any opinions, findings, conclusions or recommendations expressed in this material are those of the authors and should not be interpreted as representing the official policies, either expressed or implied, of the funding agencies
 or the U.S.~Government. The U.S.~Government is authorized to reproduce and distribute reprints for Government purposes not withstanding any copyright notation here on.

\bibliographystyle{siam}
\bibliography{bibfile}
\appendix
\section{Notation Table}
\label{sec:notation}
Table~\ref{tab:notation} lists the notation used in the paper.
\begin{table*}
    \centering
    \begin{tabular}{|c|l|}
         \hline
         \textbf{Symbol}& \textbf{Meaning} \\
         \hline
         $G$ & graph\\
         \hline
         $V$ & vertex set \\
         \hline
         $E$ & edge set\\
         \hline
         $N$ & number of vertices\\
         \hline
         $M$ & number of edges\\
         \hline
         $s$ & source vertex\\
         \hline
         $t$ & destination vertex\\
         \hline
         $p^*$ & adversary's target path  \\
         \hline
         $b$ & adversary's budget\\
         \hline
         $\calD$ & distribution of source/destination pairs \\
         \hline
         $\calP$ & distribution of target paths\\
         \hline
         $\calB$ & distribution of attacker's budget \\
         \hline
         $\calG$ & distribution of graphs a user might see \\
         \hline
         $\calX(G, w, p^*)$ & set of attacks against graph $G$ with weights $w$ that make\\
         & $p^*$ the shortest path between its endpoints\\
         \hline
         $w(e)$ & true edge weight function $w:E\rightarrow\reals_{\geq0}$\\
         \hline
         $w^\prime(e)$ & published weight function $w:E\rightarrow\reals_{\geq0}$\\
         \hline
         $c(e)$ & edge removal cost function $c:E\rightarrow\reals_{+}$\\
         \hline
         $\reals_{\geq 0}$&set of nonnegative real numbers\\
         \hline
         $\wVec$ &true edge weight vector\\
         \hline
         $\wVec^\prime$ &published edge weight vector\\
         \hline
         $\cVec$ & edge removal cost vector\\
         \hline
         $f_+$ &defender's marginal cost of overestimating user path length\\
         \hline
         $f_-$ &defender's marginal cost of underestimating user path length\\
         \hline
         $\cVec$ & vector of edge removal costs\\
         \hline
         $\Delta$ & binary vector representing edges cut \\
         \hline
         $\xVec_p$ & binary vector representing edges in path $p$ \\
         \hline
         $\ell(G, w, p)$ & length of path $p$ in graph $G$ with weights $w$\\
         \hline
         $p(G, w, s, t)$ & shortest path from $s$ to $t$ in $G$ using weights $w$\\
         \hline
         $P_{p^*}$ & set of paths with the same source and destination as $p^*$ that\\
         &are no longer than $p^*$\\
         \hline
         $P_{t}$ & set of target paths in $\calP$\\
         \hline
         $E_{p}$ & set of edges on path $p$\\
         \hline
         $(\cdot)^\top$ & matrix or vector transpose\\
         \hline
         $\langle k\rangle$ & average degree in $G$\\
         \hline 
         $\sigma_k$ & standard deviation of node degrees in $G$ \\
         \hline 
         $\kappa$ & global clustering coefficient of $G$ \\
         \hline
         $\tau$ & transitivity of $G$ \\
         \hline
         $\triangle$ & number of triangles in $G$ \\
         \hline
    \end{tabular}
    \caption{Notation used throughout the paper.}
    \label{tab:notation}
\end{table*}

\section{Computing Defender Cost}
\label{sec:defenderCost}
Psuedocode to compute the defender cost is shown in Algorithm~\ref{alg:cost}.

\begin{algorithm}
\caption{Compute defender cost}
\label{alg:cost}
\textbf{Input}: graph $G=(V, E)$, true weights $w$, published weights $w^\prime$, edge removal costs $c$, budget dist. $\mathcal{B}$, $p^*$ dist. $\mathcal{P}$, $(s,t)$ pair dist. $\mathcal{D}$, attack cost $\lambda$\\
\textbf{Output}: defender cost $L$
\begin{algorithmic}[1] 
\STATE $L\gets 0$
\FORALL{$(s,t)\in V\times V$}
\STATE $p\gets p(G, w^\prime, s, t)$  \COMMENT{all-pairs shortest path}
\STATE $d_{G, s,t}\gets\ell(G, w, p)$
\STATE $d^\prime_{G, s,t}\gets\ell(G, w^\prime, p)$
\ENDFOR
\FORALL{$p^*$ in $\mathcal{P}$}
\STATE $\ppath\gets\Pr_{P\sim\mathcal{P}}[P=p^*]$
\STATE $E^\prime\gets\textrm{attack}(G, w^\prime, p^*)$
\STATE $\btemp\gets c(E^\prime)$
\STATE $\pattack\gets\Pr_{B\sim\mathcal{B}}[B\geq\btemp]$
\IF {$\pattack > 0$}
\STATE $G^\prime\gets(V, E\setminus E^\prime)$
\FORALL{$(s,t)\in V\times V$}
\STATE $p\gets p(G^\prime, w^\prime, s, t)$  \COMMENT{all-pairs shortest path}
\STATE $d_{G^\prime, s,t}\gets\ell(G^\prime, w, p)$
\STATE $d_{G^\prime, s,t}^\prime\gets\ell(G^\prime, w^\prime, p)$
\STATE $\ppair\gets\Pr_{D\sim\mathcal{D}}[D=(s,t)]$
\STATE $L_d\gets\ppath\cdot\pattack\cdot\ppair\cdot d_{G^\prime, s, t}$
\STATE $L_e\gets\ppath\cdot\pattack\cdot\ppair\cdot\cerr\!\!\left(d_{G^\prime, s, t}, d_{G^\prime, s, t}^\prime\right)$
\STATE $L\gets L+L_d+L_e$
\ENDFOR
\ENDIF
\IF {$\pattack < 1$}
\FORALL{$(s, t)\in V\times V$}
\STATE $\ppair\gets\Pr_{(u,v)\sim\mathcal{D}}[(u,v)=(s,t)]$
\STATE $L_d\gets\ppath\cdot(1-\pattack)\cdot\ppair\cdot d_{G, s, t}$
\STATE $L_e\gets\ppath\cdot(1-\pattack)\cdot\ppair\cdot \cerr\left(d_{G, s, t}, d_{G, s, t}^\prime\right)$
\STATE $L\gets L+L_d+L_e$
\ENDFOR
\ENDIF
\STATE $L\gets L+\lambda\cdot\ppath\cdot\pattack$
\ENDFOR
\STATE \textbf{return} $L$
\end{algorithmic}
\end{algorithm}

\section{Nonlinear Program Constraints}
\label{sec:NP_constraints}
The constraints in the nonlinear program formulation serve the following purposes:
\begin{itemize}
    \item Equation~(\ref{eq:distance_cost}) defines the cost $L_d$ due to distances, weighted by $q_{uv}$---the proportion of traffic between the endpoints $u$ and $v$---and considering the expected value over target paths, hence also taking the sum after weighting by target path attack probabilities $z_{p^*}$.
    \item Equation~(\ref{eq:error_cost_total}) defines the overall cost due to differences between true and published weights, $L_e$, for travel from $u$ to $v$, weighted by the proportion of trips $q_{uv}$ and taking the expected value over all possible target paths, based on target path attack probabilities $z_{p^*}$.
    \item Equation~(\ref{eq:error_cost_single}) defines the individual error cost for a particular path selected by the user ($u$ to $v$) and a particular target path $p^*$ as being $f_+$ times the difference between the published length and the true length if the difference is positive, or $f_-$ times the absolute difference if it is negative.
    \item Equation~(\ref{eq:first_attack_const}) states that the attack vector $\Delta_{p^*}$ makes $p^*$ the shortest path between its terminal nodes.
    \item Equation~(\ref{eq:no_attack}) states that no edges will be cut if there is no target path.
    \item Equation~(\ref{eq:minDelta}) constrains the attack to be the best possible (having the smallest possible attacker cost).
    \item Equation~(\ref{eq:attack_prob}) defines $z_{p^*}$ as the target path attack probability: The probability that $p^*$ is the target path times the probability that the optimal attack for the target $p^*$ is within the attacker's budget.
    \item Equation~(\ref{eq:z0}) defines $z_\emptyset$ as the complement of all target path attack probabilities.
    \item Equation~(\ref{eq:last_attack_const}) ensures all target path attack probabilities are nonnegative.
    \item Equation~(\ref{eq:minpath}) defines $p_{uv,p^*}$ as the shortest path between $u$ and $v$ when the attacker performs an attack with target path $p^*$, according to the published weights.
    \item Equation~(\ref{eq:true_length}) defines $\ell_{uv,p^*}^{\textrm{true}}$ as the true length of $p_{uv,p^*}$ (using the true weights).
    \item Equation~(\ref{eq:obs_length}) defines $\ell_{uv,p^*}^{\textrm{obs}}$ as the observed length of $p_{uv,p^*}$ (using the published weights).
    \item Equation~(\ref{eq:diff_piece}) constrains variables denoting the positive and negative differences between the true and observed weights ($d_{uv,p^*}^{\textrm{pos}}$ and $d_{uv,p^*}^{\textrm{neg}}$, respectively) to be nonnegative.
    \item Equation~(\ref{eq:diff_diff}) constrains the difference between $d_{uv,p^*}^{\textrm{pos}}$ and $d_{uv,p^*}^{\textrm{neg}}$ to be the difference between the observed and true lengths of $p_{uv,p^*}$.
\end{itemize}

\section{Proof of Theorem~\ref{thm:connected}}
\label{sec:connectedProof}
Here we prove that the optimal attack will never result in a disconnected graph.
\begin{proof}[Proof of Theorem~\ref{thm:connected}]
If the optimal solution were to disconnect a connected graph, the resulting graph would have multiple connected components, one of which contains the target path $p^*$. Let $e$ be an edge (removed in the attack) that connects the component that includes $p^*$ to another component. Suppose $e$ is put back into the graph after the attack. Since there is a single edge connecting the two components, this will not affect the shortest path from the source $s$ to the destination $t$ of $p^*$: If a path starting at $s$ were to cross $e$ into the other component, it would have to cross back over $e$ to reach $t$, and thus would not be a shortest path. This contradicts the assumption that the proposed attack was optimal.
\end{proof}

Any attack algorithm that provides a solution to Force Path Cut may be augmented to ensure that the resulting graph is connected: If the attack results in multiple connected components, add the highest-cost edge between two connected components back to the graph until it is connected again.

\section{Proof of Theorem~\ref{thm:hardness}}
\label{sec:proof}
We begin with the graph as described in Section~\ref{subsec:zerosum}. Pseudocode for creating this graph is provided in Algorithm~\ref{alg:construction}.
\begin{algorithm}[tb]
\caption{Construct Zero-Sum Cut Defense problem from Knapsack}
\label{alg:construction}
\textbf{Input}: item values $\{\nu_1,\cdots,\nu_n\}\subset\ints_+$, item weights $\{\eta_1,\cdots,\eta_n\}\subset\ints_+$, value threshold $U$, weight threshold $H$\\
\textbf{Output}: graph $G=(V, E)$, true weights $w$, edge removal costs $c$, budget dist. $\mathcal{B}$, $p^*$ dist. $\mathcal{P}$, $(s,t)$ pair dist. $\mathcal{D}$, error costs $(f_+, f_-)$
\begin{algorithmic}[1] 
\STATE $V\gets\{u_0\}, E\gets\emptyset$
\FOR{$i\gets 1$ to $n$}
\STATE $V\gets V\cup\{u_i, \omega_i\}$
\STATE $E\gets E\cup\{\{u_{i-1}, \omega_i\},\{u_{i-1}, u_i\},\{u_i, \omega_i\}\}$
\STATE $w(\{u_{i-1}, u_i\})\gets 1$
\STATE $c(\{u_{i-1}, u_i\})\gets 1$
\STATE $w(\{u_{i-1}, \omega_i\})\gets 1$
\STATE $c(\{u_{i-1}, \omega_i\})\gets \nu_i$
\STATE $w(\{u_i, \omega_i\})\gets \eta_i $
\STATE $c(\{u_i, \omega_i\})\gets \nu_i$
\ENDFOR
\STATE $\calB\gets\begin{cases}
1 & \textrm{ if }b=U-1\\
0 & \textrm{ otherwise}
\end{cases}$
\STATE $\calP\gets\begin{cases}
1 & \textrm{ if }p^*=(u_0, u_1, \cdots, u_{n-1}, u_{n})\\
0 & \textrm{ otherwise}
\end{cases}$
\STATE $\calD\gets\begin{cases}
1 & \textrm{ if }(s, t)=(u_0, u_n)\\
0 & \textrm{ otherwise}
\end{cases}$
\STATE $f_+\gets 1$
\STATE $f_-\gets \sum_{i=1}^N{\eta_i}$
\RETURN $G=(V,E)$, $w$, $c$, $\calB$, $\calP$, $\calD$, $(f_+, f,-)$
\end{algorithmic}
\end{algorithm}

Given a graph as constructed by Algorithm~\ref{alg:construction}, consider each of the triangles $\{u_{i-1}, \omega_i, u_i\}$ individually. To make $p^*$ shortest, the attacker will have to cut either $\{u_{i-1}, \omega_i\}$ or $\{\omega_i, u_i\}$, thus incurring cost $\nu_i$, any time $\{u_{i-1}, u_i\}$ is not the shortest path from $u_{i-1}$ to $u_i$, i.e., when \begin{equation}
    w^\prime(\{u_{i-1}, u_i\})\geq w^\prime(\{u_{i-1}, \omega_i\} )+w^\prime(\{\omega_i, u_i\}).\label{eq:triangleDistance}
\end{equation}
If (\ref{eq:triangleDistance}) does not hold, then $\{u_{i-1}, u_i\}$ is the shortest path from $u_{i-1}$ to $u_i$, so the attacker will not remove any edges in triangle $\{u_{i-1}, \omega_i, u_i\}$. This observation yields the following lemma.
\begin{lemma}
The attacker will cut an edge in the triangle $\{u_{i-1}, \omega_i, u_i\}$, incurring a cost of $v_i$, if and only if (\ref{eq:triangleDistance}) holds.\label{lem:attackCondition}
\end{lemma}
In terms of the defender's cost, we show that if the defender forces the attacker to remove an edge from triangle $i$, the minimum cost is $\eta_i$.

\begin{lemma}
The optimal defender cost when a user travels from $u_{i-1}$ to $u_i$ is
\begin{equation}
    L_d+L_e=\begin{cases}
    \eta_i+1 &\textrm{ if (\ref{eq:triangleDistance}) holds}\\
    1 &\textrm{ otherwise}
\end{cases}.
\end{equation}
Furthermore, the optimal solution only perturbs edges $\{u_{i-1},u_i\}$, the edges along $p^*$.\label{lem:defenseCost}
\end{lemma}
\begin{proof}
In the context of a Zero-Sum Cut Defense problem constructed by Algorithm~\ref{alg:construction}, it is always possible to reduce the probability of attack to zero. Thus, we will focus on the case where no attack occurs. Let $\delta$ be difference between the true and published distance for the path $(u_{i-1}, u_i)$ (i.e., the published distance is $1+\delta$) and $\epsilon$ be the same for path $(u_{i-1}, \omega_i, u_i)$ (published distance $1+\eta_i+\epsilon$). Recall that the marginal increase in cost for publishing a longer distance than the user experiences is 1, while the marginal cost increase for publishing shorter distances is $H^\prime$. If equality holds in (\ref{eq:triangleDistance}), either path may be taken, so we consider all possible costs. Note that, in this case, $\delta=\eta_i+\epsilon$.
\begin{itemize}
    \item If $\delta\geq 0$ and $\epsilon\geq 0$, then the cost of taking $(u_{i-1}, u_i)$ is $1+\delta=1+\eta_i+\epsilon\geq1+\eta_i$. The cost of taking $(u_{i-1}, \omega_i, u_i)$ is $1+\eta_i+\epsilon\geq1+\eta_i$.
    \item If $\delta \geq 0$ and $\epsilon < 0$, then the cost of taking $(u_{i-1}, u_i)$ is $1+\delta=1+\eta_i+\epsilon<1+\eta_i$, but the cost of taking 
    $(u_{i-1}, \omega_i, u_i)$ is $1+\eta_i+H^\prime|\epsilon|>1+\eta_i$.
    \item If $\delta < 0$ and $\epsilon < 0$, then the cost of taking $(u_{i-1}, u_i)$ is $1+H^\prime|\delta|$, which may be smaller than $1+\eta_i$, but the cost of taking 
    $(u_{i-1}, \omega_i, u_i)$ is $1+\eta_i+H^\prime|\epsilon|>1+\eta_i$.
\end{itemize}
The first case is the only one where the cost is never strictly above $1+\eta_i$, so this cases optimizes the worst-case cost. It achieves the lower bound only when $\epsilon=0$ and $\delta=\eta_i$.

If equality does not hold in (\ref{eq:triangleDistance}), then the path $(u_{i-1}, \omega_i, u_i)$ is taken from $u_{i-1}$ to $u_i$. The cost is minimized when no edge weight along this path is perturbed, and the cost is $1+\eta_i$.

If (\ref{eq:triangleDistance}) does not hold, the user will take path $(u_{i-1}, u_i)$. As in the previous case, cost is minimized when no perturbation takes place, yielding a cost of 1.
\end{proof}

We now introduce the full reduction from Knapsack to Zero-Sum Cut Defense. We start by creating the Zero-Sum Cut Defense problem using Algorithm~\ref{alg:construction}. We run Zero-Sum Cut Defense on this graph and check each of its triangles $\{u_{i-1}, \omega_i, u_i\}$. If the published weight of edge $\{u_{i-1}, \omega_i, u_i\}$ is greater than its true weight, we add the item to the knapsack. If the total weight of the items is no more than the prescribed threshold, we return true; the answer to the decision question is yes. Otherwise, we return false. Pseudocode for the reduction is provided in Algorithm~\ref{alg:reduction}.
\begin{algorithm}[tb]
\caption{Reduction from Knapsack to  Zero-Sum Cut Defense}
\label{alg:reduction}
\textbf{Input}: item values $\{\nu_1,\cdots,\nu_n\}\subset\ints_+$, item weights $\{\eta_1,\cdots,\eta_n\}\subset\ints_+$, value threshold $U\in\ints_+$, weight threshold $H\in\ints_+$\\
\textbf{Output}: Boolean value indicating answer to Knapsack
\begin{algorithmic}[1] 
\STATE $G$,$w$, $c$, $\mathcal{B}$, $\mathcal{P}$, $\mathcal{D}$, $(f_+, f_-)\gets$ Algorithm~\ref{alg:construction}($\{\nu_1,\cdots,\nu_n\}$, $\{\eta_1,\cdots,\eta_n\}$, $U$, $H$)
\STATE $w^\prime\gets$ Zero-Sum Cut Defense($G$,$w$, $c$, $\mathcal{B}$, $\mathcal{P}$, $\mathcal{D}$, $(f_+, f_-)$)
\STATE $\eta_{\textrm{total}}\gets0$
\FOR{$i\gets 1$ to $n$}
\IF{$w^\prime(\{(u_{i-1}, u_i)\})>w(\{(u_{i-1}, u_i)\})$}
\STATE $\eta_{\textrm{total}}\gets \eta_{\textrm{total}}+\eta_i$
\ENDIF
\ENDFOR
\IF{$\eta_{\textrm{total}}\leq H$}
\RETURN \TRUE
\ELSE
\RETURN \FALSE
\ENDIF
\end{algorithmic}
\end{algorithm}

This reduction solves the Knapsack problem, as we formally prove now.
\begin{lemma}
Algorithm~\ref{alg:reduction} returns true if and only if there is a subset of items whose weights sum to no more than $H$ and value sums to at least $U$.\label{lem:reduction}
\end{lemma}
\begin{proof}
If there is a subset $X\subset\{1,\cdots,n\}$ where $\sum_{i\in X}{\nu_i}\geq U$ and $\sum_{i\in X}{\eta_i}\leq H$, then one solution to Zero-Sum Cut Defense is to increase the edge weight on $\{u_{i-1},u_i\}$ by at least $\eta_i$ for all $i\in X$. This will result in a case where the attacker must spend at least $U$ to be effective: By Lemma~\ref{lem:attackCondition}, the attacker must remove an edge with removal cost $\nu_i$ for all $i\in X$, yielding a required attack budget of at least $U$ by our assumption, thus exceeding the budget constructed in Algorithm~\ref{alg:construction}. By Lemma~\ref{lem:defenseCost}, the cost to the defender of each perturbation will increase by $\eta_i$, which we also assume will sum to at most $H$. Since such a defense exists, the minimum-cost defense must have cost no more than $H$. Thus, if there is a solution to Knapsack, Algorithm~\ref{alg:reduction} will return true.

Now suppose that there is no such set, but Algorithm~\ref{alg:reduction} returns true. For the attack probability to be zero, the optimal attack must cost at least $U$. Let $E^\prime$ be the edges cut in an optimal attack. All edges will be of the form $\{u_{i-1},\omega_i\}$ or $\{\omega_i,u_i\}$, since all edges $\{u_{i-1}, u_i\}$ are part of $p^*$. There will never be an optimal attack that removes both $\{u_{i-1},\omega_i\}$ and $\{\omega_i,u_i\}$: this would increase the cost of the attack and not remove any competing paths. Thus, let $Y\subset\{1,\cdots,n\}$ be the indices of the triangles that are cut, i.e., $$i\in Y\iff\{u_{i-1},\omega_i\}\in E^\prime\textrm{ or }\{\omega_i,u_i\}\in E^\prime.$$
Since the defense minimizes the attack probability, we have $\sum_{i\in Y}{\nu_i}\geq U$. By Lemma~\ref{lem:attackCondition}, we know that these edges will be cut only if (\ref{eq:triangleDistance}) holds, i.e., it holds for all $i\in Y$ and not for any $i\in\{1,\cdots,n\}\setminus Y$. From Lemma~\ref{lem:defenseCost}, we know that the cost is optimized in this case by perturbing $w(\{u_{i-1},u_i\})$, increasing it by $\eta_i$. Thus, $\eta_{\textrm{total}}=\sum_{i\in Y}{\eta_i}$. Algorithm~\ref{alg:reduction} returns true only if $\eta_{\textrm{total}}\leq H$. However, this contradicts our assumption that there is no solution to Knapsack: $Y$ is a set of items with weight no more than $H$ and value at least $U$. Thus, when there is no solution to Knapsack, Algorithm~\ref{alg:reduction} will return false.
\end{proof}

Using these intermediate results, we now prove the original theorem.
\begin{proof}[Proof of Theorem~\ref{thm:hardness}]
Lemma~\ref{lem:reduction} shows that Algorithm~\ref{alg:reduction} is a reduction from Knapsack to Zero-Sum Cut Defense. Constructing the graph includes $n$ iterations of a loop, each of which performs a sub-linear amount of work. Thus, the loop completes in polynomial time. Creating the distributions requires at most linear time, and computing the marginal cost of error takes linear time. Thus, Algorithm~\ref{alg:reduction} is a polynomial-time reduction from Knapsack to Zero-Sum Cut Defense. Since Knapsack is an NP-hard problem, this implies that Zero-Sum Cut Defense is NP-hard as well.
\end{proof}
\section{Dataset Features}
\label{sec:datasets}
We use the following datasets in our experiments:
\begin{itemize}
    \item ER: Erd\H{o}s--R\'{e}nyi random graphs where each pair of nodes shares an edge with probability $0.048$.
    \item BA: Barab\'{a}si--Albert graphs where incoming nodes attach with 6 edges.
    \item WS: Watts--Strogatz graphs with average degree 12 and rewiring probability $0.05$.
    \item SBM: Stochastic blockmodel (SBM) graphs where nodes are separated into a 200-node cluster with internal connection probability $0.06$ and a 50-node cluster with connection probability $0.2$, and two nodes in different clusters share an edge with probability $0.005$.
    \item USAIR: Nodes are 500 airports in the United States, edge weights are the total numbers of seats on all flights between them. Available at \url{https://toreopsahl.com/datasets/\#usairports}.
    \item UKMET: Nodes are metro stops and weights are the average travel time between the stops. We take the largest strongly connected component and make it undirected by averaging the weights in each direction. Available at \url{https://datadryad.org/stash/dataset/doi:10.5061/dryad.pc8m3}.
    \item HT: Nodes are attendees at the ACM Hypertext 2009 conference, and edges are the number of interactions between people over the course of the two days. Available at \url{http://www.sociopatterns.org/datasets/hypertext-2009-dynamic-contact-network/}.
    \item AS: Nodes are routers and edges are connections between them (dated 11 Dec. 1999). Available at \url{http://snap.stanford.edu/data/as-733.html}.
\end{itemize}
Statistics of the datasets are provided in Table~\ref{table:syn_net_prop}. Code used in the experiments, including the code used for dataset preparation, is available at \url{https://github.com/bamille1/PATHDEFENSE}.

\begin{table*}[!ht]
\renewcommand{\arraystretch}{1.0}
\centering
\begin{tabular}{|l||c|c|c|c|c|c|c|} 
\hline
\!Network & $|V|$ & $|E|$ & $\langle k\rangle$ & $\sigma_k$ & $\kappa$ & $\tau$ & $\triangle$ \\
\hline 
\!ER	&	 	$250$ &			$1498.3$ &		$11.986 $ & 	$3.352$ 	& $0.048$ 					& $0.047$ 	  &	$280.6$ \\ 
    &             &         $\pm32.885$ &   $\pm0.263$ &    $\pm0.144$	&  $\pm0.003$               & $\pm0.002$  & $\pm18.645$         \\
\hline
\!BA	&	 	$250$ & 		$1464.0 $ &		$11.712 $ & 	$9.345$ 	& $0.122$ 					& $0.099$ 	  &	$878.1 $ \\ 
    &             &         $\pm0.0$ &      $\pm0.0$ &      $\pm0.272$	&  $\pm0.006$               & $\pm0.003$  & $\pm36.585$                    \\
\hline
\!WS	&	 	$250$ & 		$1500.0 $ &		$12.0 $ & 		$0.774$ 	& $0.586$ 					& $0.582$ 	  &	$3214.4 $ \\ 
    &             &         $\pm0.0$ &      $\pm0.0$ &      $\pm0.063$	&  $\pm0.013$               & $\pm0.013$  & $\pm71.256$              \\
\hline
\!SBM	&	 	$250$ & 		$1479.2$ &  	$11.834 $ & 	$3.264$ 	& $0.079$ 					& $0.073$ 	  &	$424.5 $ \\ 
    &             &         $\pm36.13$ &    $\pm0.289$ &    $\pm0.167$	&  $\pm0.004$               & $\pm0.005$  & $\pm47.124$                  \\
\hline
\hline
\!AS	&	 	$1477$ 	&		$3142$ &				$4.254 $ & 			$15.814$ & 	  $0.242$ & 		$0.038$ &		$2530$  \\ 
\hline
\!HT	&	 	$113$ 	&		$2196 $ &				$38.867 $ & 			$18.350$ & 	  $0.534$ & 		$0.495$ &		$16867$  \\ 
\hline
\!USAIR	&	$500$ 	&		$2980 	$ &				$11.92 $ & 				$22.338$ & 	  $0.617$ & 		$0.351$ &		$41583$  \\ 
\hline
\!UKMET\!\! & $298$ & $349$ & $2.342$ & $1.031$ & $0.042$ & $0.086$ & $18$  \\
\hline

\end{tabular}
\caption{Properties of the synthetic and real networks used in our experiments. For each random graph model, we generate 10 networks. Note that the number of edges across the different synthetic networks is $\approx1500$. The table shows the average degree ($\langle k\rangle$), standard deviation of the degree ($\sigma_k$), average clustering coefficient ($\kappa$), transitivity ($\tau$), and number of triangles ($\triangle$). Each graph has a single connected component. The $\pm$ values show the standard deviation across 10 runs of each random graph model.}
\label{table:syn_net_prop}
\end{table*}

\section{Extended Experimental Results}
\label{sec:moreResults}
Here we provide additional experimental results. Results where all the target paths have the same endpoints are shown in Fig.~\ref{fig:sameResultsSynth} and Fig.~\ref{fig:sameResultsReal} for synthetic and real graphs, respectively. Results where each target's endpoints are chosen independently of one another are shown in Fig.~\ref{fig:differentResultsSynth} (synthetic) and Fig.~\ref{fig:differentResultsReal} (real).

\begin{figure*}
    \centering
    \includegraphics[width=\textwidth]{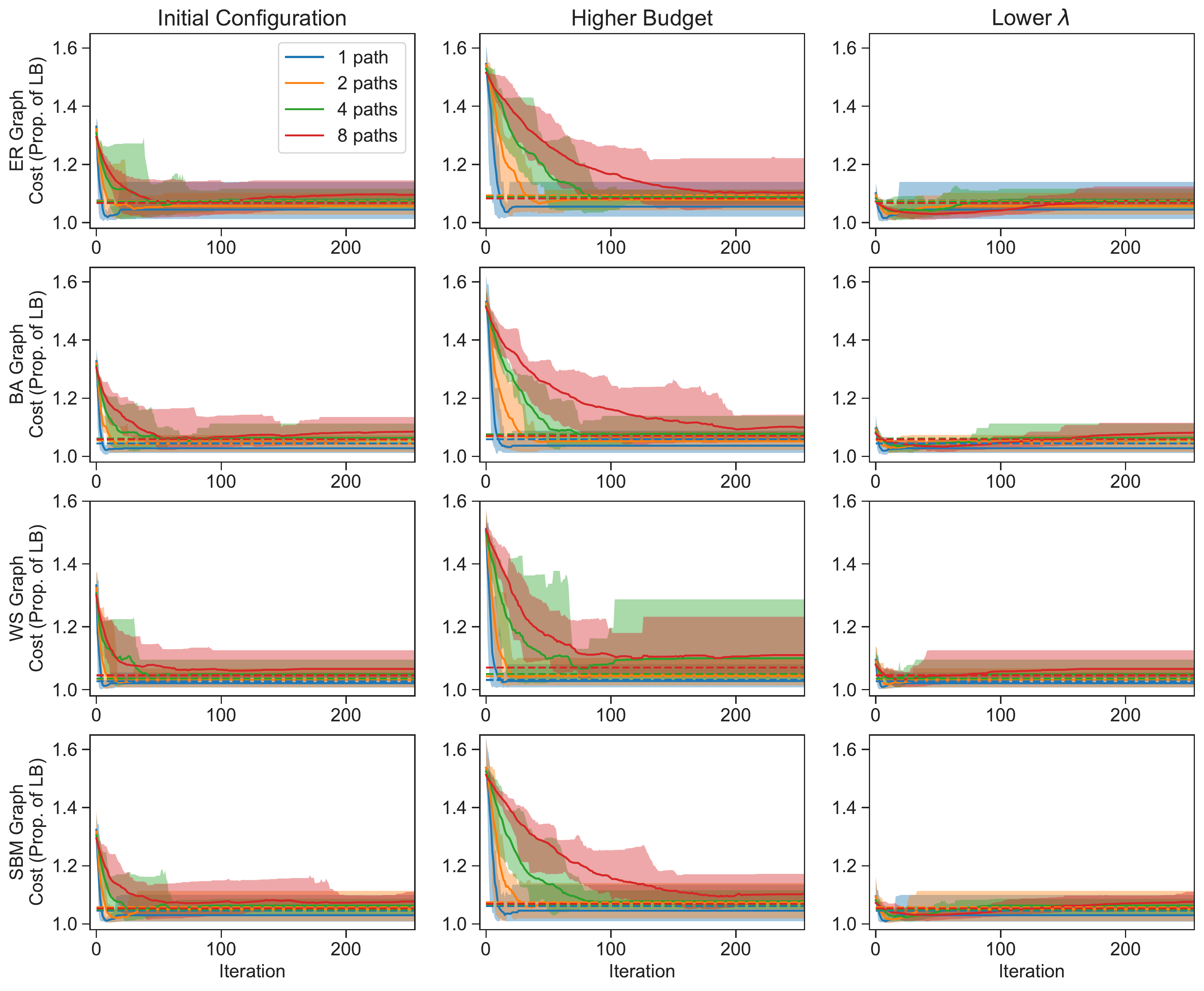}
    \caption{Cost of \texttt{PATHDEFENSE} for all synthetic datasets under various conditions. Each row plots the results for a different dataset, and results are shown for the original budget and $\lambda$ (left column), when the attacker budget is doubled (center column), and when the cost of attacker success is reduced by five times (right column). All target paths use the same terminal nodes. Plots include the average cost (solid line) and the cost range across trials (shaded area), as well as the average zero-sum result (dash line). All cases follow a similar trajectory: an initial decrease in cost followed by a mild increase.}
    \label{fig:sameResultsSynth}
\end{figure*}
\begin{figure*}
    \centering
    \includegraphics[width=\textwidth]{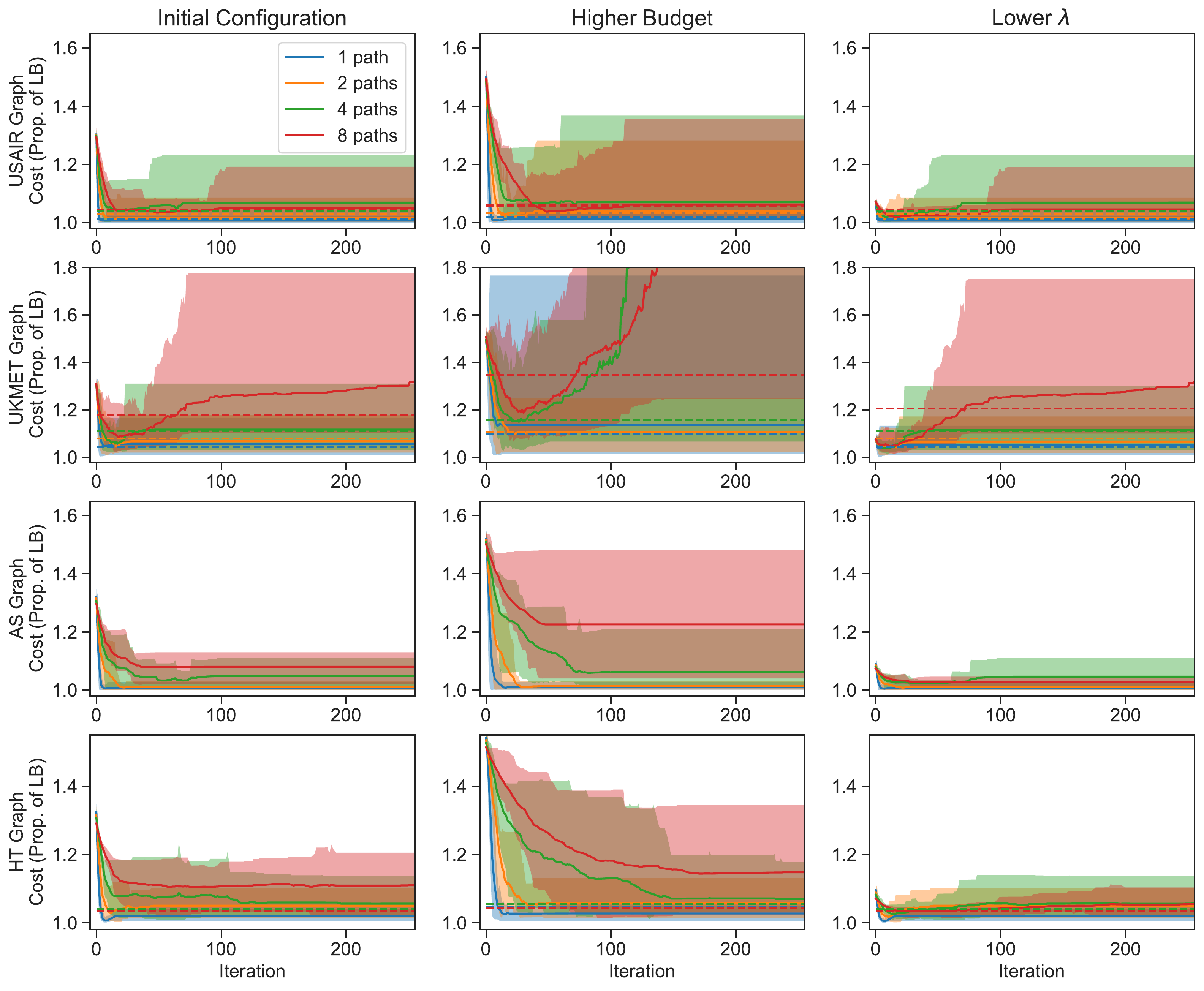}
    \caption{Cost of \texttt{PATHDEFENSE} for all real datasets under various conditions. Each row plots the results for a different dataset, and results are shown for the original budget and $\lambda$ (left column), when the attacker budget is doubled (center column), and when the cost of attacker success is reduced by five times (right column). All target paths use the same terminal nodes. Plots include the average cost (solid line) and the cost range across trials (shaded area), as well as the average zero-sum result (dash line), with the exception of the AS graph, where the zero-sum procedure did not complete in 24 hours and are omitted. Running PATHDEFENSE on both transportation graphs yield a decrease in cost followed by an increase, while the cost with the computer and social networks decrease and become stable when the cost of adversary success $\lambda$ is high. When $\lambda$ is lowered, PATHDEFENSE yields a cost increase after about 100 iterations, with the defender eventually selecting the result of an early iteration.}
    \label{fig:sameResultsReal}
\end{figure*}
\begin{figure*}
    \centering
    \includegraphics[width=\textwidth]{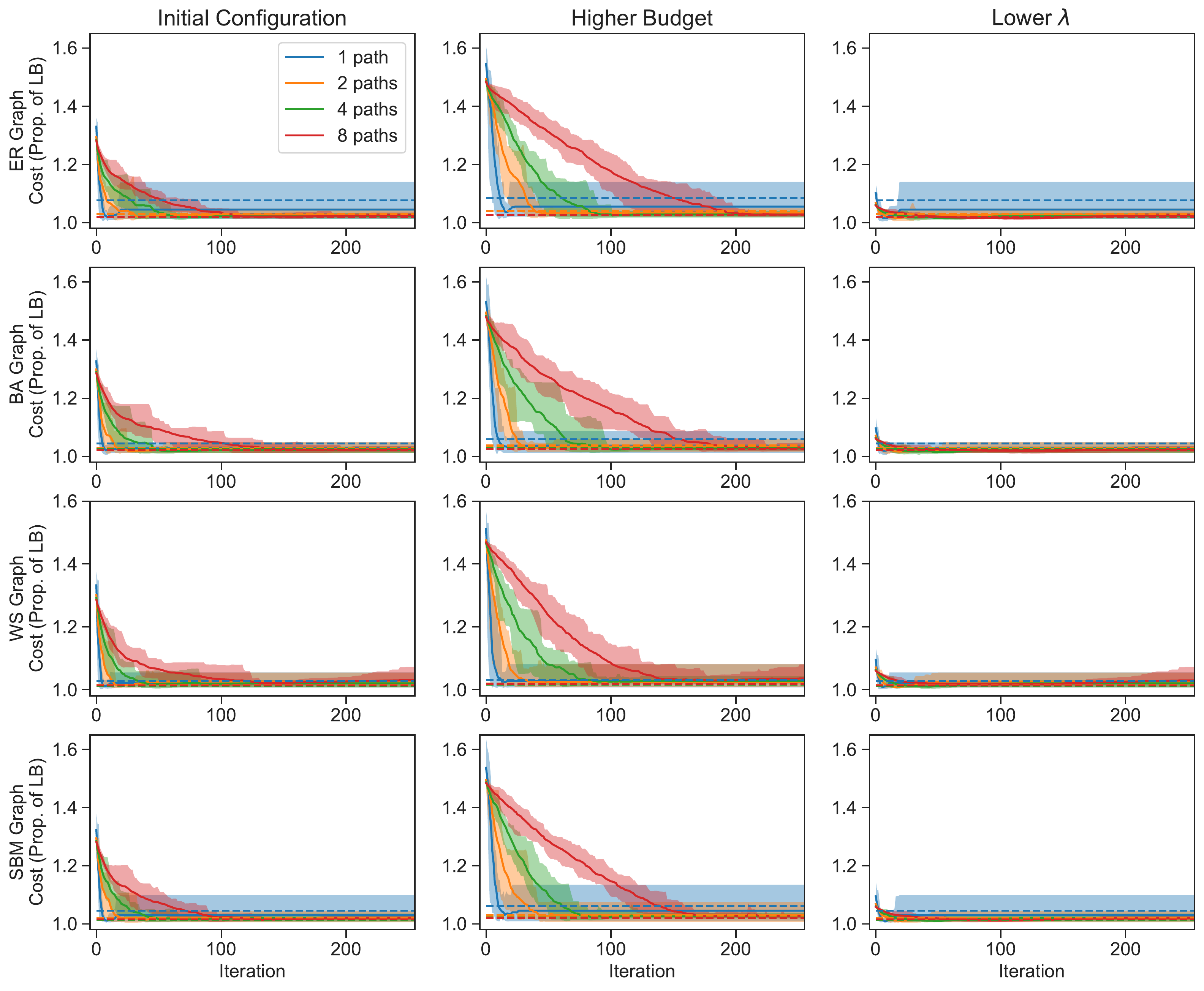}
    \caption{Cost of \texttt{PATHDEFENSE} for all synthetic datasets under various conditions. Each row plots the results for a different dataset, and results are shown for the original budget and $\lambda$ (left column), when the attacker budget is doubled (center column), and when the cost of attacker success is reduced by five times (right column). All target paths use different terminal nodes. Plots include the average cost (solid line) and the cost range across trials (shaded area), as well as the average zero-sum result (dash line). The variation in cost is lower than in the case where all target paths have the same endpoints, likely due to less interference between perturbations for different targets (i.e., a perturbation making one target more difficult to attack is less likely to make another target easier to attack if the two paths do not share start and end points).}
    \label{fig:differentResultsSynth}
\end{figure*}
\begin{figure*}
    \centering
    \includegraphics[width=\textwidth]{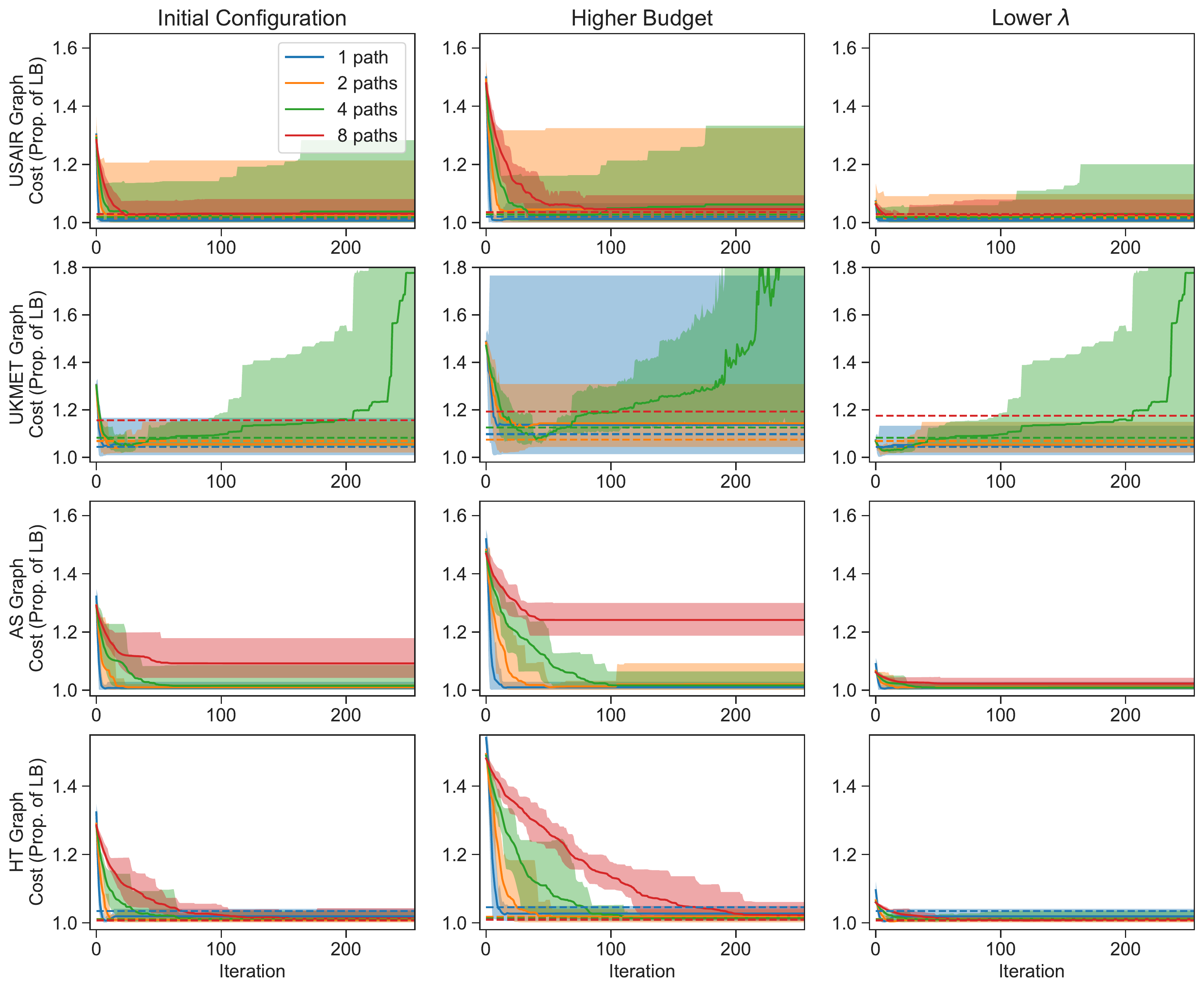}
    \caption{Cost of \texttt{PATHDEFENSE} for all real datasets under various conditions. Each row plots the results for a different dataset, and results are shown for the original budget and $\lambda$ (left column), when the attacker budget is doubled (center column), and when the cost of attacker success is reduced by five times (right column). All target paths use different terminal nodes. Plots include the average cost (solid line) and the cost range across trials (shaded area), as well as the average zero-sum result (dash line), with the exception of the AS graph, where the zero-sum procedure did not complete in 24 hours and are omitted. As in synthetic graphs, the cost variation is lower than in the case where targets share endpoints, though we see some similar qualitative phenomena (large increase in cost for USAIR and UKMET).}
    \label{fig:differentResultsReal}
\end{figure*}

\section{Baseline Comparison}
\label{sec:baseline}
As a baseline, we consider adding large weights only to the possible target paths. We begin by computing the sum of all weights in the graph, $w_{\textrm{all}}=\sum_{e\in E}{w(e)}$. For each possible target path $p^*$, we consider all edges in $E_{p^*}$. We change the weight of edge $e$ to $w^\prime(e)=\max_{p^*|e\in{E_{p^*}}}\frac{w_{\textrm{all}}}{|E_{p^*}|}$. This is guaranteed to make each $p^*$ longer than any existing path in the graph. We refer to this baseline method as \texttt{BigWeight}.

We consider three possible attacks the adversary may use: (1) the \texttt{PATHATTACK} algorithm used in the experiments in Section~\ref{sec:experiments}, (2) the greedy baseline \texttt{GreedyCost} used as a baseline for \texttt{PATHATTACK} (greedily removing the lowest-cost edge of the shortest path until $p^*$ is shortest)~\cite{Miller2021}, and (3) the optimal solution, found using a mixed-integer program solver (thus not guaranteed to be solved in polynomial time). 

Results providing a comparison to the \texttt{BigWeight} baseline are provided in Table~\ref{tab:comparison}. Regardless of the attack method being used by the adversary, the cost of defense is quite similar. Looking deeper into the results, we see that there tends to be a small increase in $L_d$ as more sophisticated methods are used (i.e., smallest for \texttt{GreedyCost}, largest for the optimal attack), though these increases are small enough that they are not statistically significant when aggregating across trials. In all cases, \texttt{PATHDEFENSE} outperforms the \texttt{BigWeight} baseline. The difference is more prominent in the real datasets than the synthetic ones, and is particularly large in the AS graph. These graphs have more edges that cannot avoid being traversed, so using \texttt{PATHDEFENSE} to identify specific edges that minimize the increased cost has a greater benefit. 
 Results for cases where the terminal nodes of the target paths differ are available in Table~\ref{tab:comparison_different}.
\begin{table*}[t]
    \centering
    \caption{Comparison of defender costs with \texttt{PATHDEFENSE} and the \texttt{BigWeight} baseline when all target paths share terminal nodes. Results are shown for all datasets with 1, 2, 4, and 8 target paths, and consider cases where the attacker uses the \texttt{GreedyCost} heuristic, the \texttt{PATHATTACK} approximation algorithm, or uses a mixed-integer program solver to compute the attack. Each cost is reported in terms of the average and standard error over 10 trials. There is typically not a substantial change in defender cost between the various attack methods, as the defenses typically drive the attack cost to zero and have a relatively small impact on distances traveled. Regardless of the attack used, \texttt{PATHDEFENSE} yields a lower cost than the baseline method in all cases.}
    \label{tab:comparison}
    {\small\begin{tabular}{cccccccc}
    \toprule
    & &\multicolumn{2}{c}{\texttt{GreedyCost}}&\multicolumn{2}{c}{\texttt{PATHATTACK}}&\multicolumn{2}{c}{optimal attack}\\
    graph & $|P_t|$ & \texttt{PATHDEFENSE} &\texttt{BigWeight} & \texttt{PATHDEFENSE} &\texttt{BigWeight} & \texttt{PATHDEFENSE} &\texttt{BigWeight} \\
    \midrule
ER&1
&$19.51\pm0.16$
&$22.80\pm0.33$
&$19.60\pm0.16$
&$22.80\pm0.33$
&$19.62\pm0.14$
&$22.80\pm0.33$
\\
ER&2
&$20.24\pm0.21$
&$23.61\pm0.29$
&$20.54\pm0.23$
&$23.61\pm0.29$
&$20.40\pm0.22$
&$23.61\pm0.29$
\\
ER&4
&$21.20\pm0.19$
&$24.01\pm0.20$
&$21.36\pm0.21$
&$24.01\pm0.20$
&$21.30\pm0.19$
&$24.01\pm0.20$
\\
ER&8
&$22.23\pm0.17$
&$24.41\pm0.16$
&$22.56\pm0.31$
&$24.41\pm0.16$
&$22.66\pm0.30$
&$24.41\pm0.16$
\\
BA&1
&$18.37\pm0.39$
&$20.53\pm0.38$
&$18.47\pm0.43$
&$20.53\pm0.38$
&$18.50\pm0.42$
&$20.53\pm0.38$
\\
BA&2
&$18.53\pm0.35$
&$21.03\pm0.29$
&$18.85\pm0.37$
&$21.03\pm0.29$
&$18.77\pm0.36$
&$21.03\pm0.29$
\\
BA&4
&$19.10\pm0.31$
&$21.57\pm0.27$
&$19.39\pm0.28$
&$21.57\pm0.27$
&$19.48\pm0.26$
&$21.57\pm0.27$
\\
BA&8
&$19.71\pm0.28$
&$22.15\pm0.19$
&$19.94\pm0.26$
&$22.15\pm0.19$
&$19.94\pm0.26$
&$22.15\pm0.19$
\\
WS&1
&$26.53\pm0.43$
&$30.38\pm0.90$
&$26.64\pm0.43$
&$30.38\pm0.90$
&$26.67\pm0.43$
&$30.38\pm0.90$
\\
WS&2
&$26.59\pm0.42$
&$31.06\pm0.92$
&$26.76\pm0.42$
&$31.06\pm0.92$
&$26.92\pm0.42$
&$31.06\pm0.92$
\\
WS&4
&$27.20\pm0.48$
&$32.02\pm0.85$
&$27.46\pm0.48$
&$32.02\pm0.85$
&$27.50\pm0.47$
&$32.02\pm0.85$
\\
WS&8
&$27.68\pm0.56$
&$32.80\pm0.87$
&$27.89\pm0.51$
&$32.80\pm0.87$
&$28.03\pm0.54$
&$32.80\pm0.87$
\\
SBM&1
&$20.44\pm0.44$
&$23.06\pm0.33$
&$20.48\pm0.44$
&$23.06\pm0.33$
&$20.46\pm0.43$
&$23.06\pm0.33$
\\
SBM&2
&$20.66\pm0.41$
&$24.10\pm0.22$
&$20.76\pm0.41$
&$24.10\pm0.22$
&$20.74\pm0.43$
&$24.10\pm0.22$
\\
SBM&4
&$21.57\pm0.28$
&$24.50\pm0.25$
&$21.79\pm0.23$
&$24.50\pm0.25$
&$21.90\pm0.23$
&$24.50\pm0.25$
\\
SBM&8
&$22.50\pm0.24$
&$24.98\pm0.27$
&$22.98\pm0.23$
&$24.98\pm0.27$
&$23.00\pm0.23$
&$24.98\pm0.27$
\\
USAIR&1
&$11.98\pm0.33$
&$20.57\pm5.46$
&$11.97\pm0.32$
&$20.57\pm5.46$
&$11.97\pm0.32$
&$20.57\pm5.46$
\\
USAIR&2
&$11.64\pm0.35$
&$20.63\pm5.92$
&$11.67\pm0.38$
&$20.63\pm5.92$
&$11.65\pm0.36$
&$20.63\pm5.92$
\\
USAIR&4
&$11.33\pm0.44$
&$18.59\pm4.63$
&$11.45\pm0.48$
&$18.59\pm4.63$
&$11.40\pm0.50$
&$20.90\pm6.41$
\\
USAIR&8
&$11.22\pm0.43$
&$19.73\pm5.68$
&$11.20\pm0.44$
&$19.73\pm5.68$
&$11.31\pm0.43$
&$30.10\pm12.50$
\\
UKMET&1
&$11.28\pm0.60$
&$15.29\pm1.33$
&$11.49\pm0.62$
&$15.30\pm1.32$
&$11.54\pm0.61$
&$15.43\pm1.34$
\\
UKMET&2
&$11.27\pm0.60$
&$16.34\pm1.26$
&$11.42\pm0.61$
&$16.23\pm1.24$
&$11.54\pm0.61$
&$16.28\pm1.24$
\\
UKMET&4
&$11.52\pm0.61$
&$18.11\pm1.17$
&$11.53\pm0.61$
&$18.05\pm1.13$
&$11.84\pm0.62$
&$18.20\pm1.15$
\\
UKMET&8
&$11.64\pm0.60$
&$19.33\pm1.06$
&$11.96\pm0.72$
&$19.37\pm1.03$
&$12.08\pm0.64$
&$19.64\pm1.06$
\\
AS&1
&$25.60\pm0.65$
&$66.32\pm18.85$
&$25.62\pm0.66$
&$66.32\pm18.85$
&$25.63\pm0.67$
&$66.32\pm18.85$
\\
AS&2
&$25.39\pm0.67$
&$98.57\pm18.81$
&$25.44\pm0.68$
&$98.57\pm18.81$
&$25.42\pm0.66$
&$98.57\pm18.81$
\\
AS&4
&$25.63\pm0.57$
&$101.99\pm17.34$
&$25.52\pm0.57$
&$101.99\pm17.34$
&$25.60\pm0.59$
&$101.99\pm17.34$
\\
AS&8
&$26.71\pm0.72$
&$102.56\pm15.09$
&$26.90\pm0.50$
&$102.57\pm15.09$
&$26.77\pm0.52$
&$102.58\pm15.09$
\\
HT&1
&$21.07\pm0.36$
&$23.13\pm0.49$
&$21.03\pm0.35$
&$23.13\pm0.49$
&$21.07\pm0.36$
&$23.13\pm0.49$
\\
HT&2
&$20.88\pm0.32$
&$23.52\pm0.53$
&$20.97\pm0.33$
&$23.52\pm0.53$
&$20.95\pm0.32$
&$23.52\pm0.53$
\\
HT&4
&$21.24\pm0.35$
&$23.62\pm0.52$
&$21.47\pm0.39$
&$23.62\pm0.52$
&$21.43\pm0.41$
&$23.62\pm0.52$
\\
HT&8
&$21.95\pm0.31$
&$23.45\pm0.43$
&$22.94\pm0.55$
&$23.45\pm0.43$
&$22.94\pm0.54$
&$23.45\pm0.43$
\\
    \bottomrule
    \end{tabular}}%
\end{table*}

\begin{table*}[t]
    \centering
    \caption{Comparison of defender costs with \texttt{PATHDEFENSE} and the \texttt{BigWeight} baseline when the target paths have different terminal nodes. Results are shown for all datasets with 1, 2, 4, and 8 target paths, and consider cases where the attacker uses the \texttt{GreedyCost} heuristic, the \texttt{PATHATTACK} approximation algorithm, or uses a mixed-integer program solver to compute the attack. Each cost is reported in terms of the average and standard error over 10 trials. As with the case where all target paths share terminals, \texttt{PATHDEFENSE} always outperforms the \texttt{BigWeight} baseline. Note that the case with the attacker using \texttt{PATHATTACK} on the UKMET graph with 8 target path did not complete in 24 hours.}
    \label{tab:comparison_different}
    {\small\begin{tabular}{cccccccc}
    \toprule
    & &\multicolumn{2}{c}{\texttt{GreedyCost}}&\multicolumn{2}{c}{\texttt{PATHATTACK}}&\multicolumn{2}{c}{optimal attack}\\
    graph & $|P_t|$ & \texttt{PATHDEFENSE} &\texttt{BigWeight} & \texttt{PATHDEFENSE} &\texttt{BigWeight} & \texttt{PATHDEFENSE} &\texttt{BigWeight} \\
    \midrule
ER&1
&$19.51\pm0.16$
&$22.80\pm0.33$
&$19.60\pm0.16$
&$22.80\pm0.33$
&$19.62\pm0.14$
&$22.80\pm0.33$
\\
ER&2
&$21.62\pm0.09$
&$23.20\pm0.11$
&$21.82\pm0.14$
&$23.20\pm0.11$
&$21.74\pm0.08$
&$23.20\pm0.11$
\\
ER&4
&$22.51\pm0.09$
&$23.54\pm0.09$
&$22.75\pm0.08$
&$23.54\pm0.09$
&$22.70\pm0.08$
&$23.54\pm0.09$
\\
ER&8
&$22.78\pm0.08$
&$23.71\pm0.09$
&$23.05\pm0.08$
&$23.71\pm0.09$
&$23.00\pm0.08$
&$23.71\pm0.09$
\\
BA&1
&$18.37\pm0.39$
&$20.53\pm0.38$
&$18.47\pm0.43$
&$20.53\pm0.38$
&$18.50\pm0.42$
&$20.53\pm0.38$
\\
BA&2
&$19.66\pm0.30$
&$21.13\pm0.32$
&$19.78\pm0.29$
&$21.13\pm0.32$
&$19.82\pm0.30$
&$21.13\pm0.32$
\\
BA&4
&$20.38\pm0.19$
&$21.55\pm0.14$
&$20.56\pm0.18$
&$21.55\pm0.14$
&$20.60\pm0.19$
&$21.55\pm0.14$
\\
BA&8
&$20.89\pm0.12$
&$21.90\pm0.09$
&$21.17\pm0.11$
&$21.90\pm0.09$
&$21.18\pm0.11$
&$21.90\pm0.09$
\\
WS&1
&$26.53\pm0.43$
&$30.38\pm0.90$
&$26.64\pm0.43$
&$30.38\pm0.90$
&$26.67\pm0.43$
&$30.38\pm0.90$
\\
WS&2
&$29.55\pm0.37$
&$32.17\pm0.36$
&$29.71\pm0.38$
&$32.17\pm0.36$
&$29.70\pm0.38$
&$32.17\pm0.36$
\\
WS&4
&$31.16\pm0.35$
&$33.41\pm0.34$
&$31.35\pm0.36$
&$33.41\pm0.34$
&$31.36\pm0.36$
&$33.41\pm0.34$
\\
WS&8
&$31.93\pm0.29$
&$34.45\pm0.31$
&$32.34\pm0.28$
&$34.45\pm0.31$
&$32.17\pm0.29$
&$34.45\pm0.31$
\\
SBM&1
&$20.44\pm0.44$
&$23.06\pm0.33$
&$20.48\pm0.44$
&$23.06\pm0.33$
&$20.46\pm0.43$
&$23.06\pm0.33$
\\
SBM&2
&$22.70\pm0.27$
&$24.16\pm0.33$
&$22.76\pm0.26$
&$24.16\pm0.33$
&$22.74\pm0.27$
&$24.16\pm0.33$
\\
SBM&4
&$23.95\pm0.23$
&$25.00\pm0.25$
&$24.07\pm0.25$
&$25.00\pm0.25$
&$24.12\pm0.24$
&$25.00\pm0.25$
\\
SBM&8
&$24.24\pm0.16$
&$25.19\pm0.16$
&$24.39\pm0.16$
&$25.19\pm0.16$
&$24.42\pm0.17$
&$25.19\pm0.16$
\\
USAIR&1
&$11.98\pm0.33$
&$20.57\pm5.46$
&$11.97\pm0.32$
&$20.57\pm5.46$
&$11.97\pm0.32$
&$20.57\pm5.46$
\\
USAIR&2
&$12.62\pm0.51$
&$46.91\pm8.35$
&$12.86\pm0.59$
&$54.49\pm8.51$
&$12.85\pm0.58$
&$54.49\pm8.51$
\\
USAIR&4
&$12.90\pm0.33$
&$46.04\pm6.66$
&$13.07\pm0.40$
&$48.83\pm5.89$
&$13.07\pm0.40$
&$49.01\pm6.55$
\\
USAIR&8
&$13.61\pm0.27$
&$48.36\pm6.51$
&$13.78\pm0.29$
&$49.23\pm6.09$
&$13.88\pm0.29$
&$44.70\pm4.37$
\\
UKMET&1
&$11.28\pm0.60$
&$15.29\pm1.33$
&$11.49\pm0.62$
&$15.30\pm1.32$
&$11.54\pm0.61$
&$15.43\pm1.34$
\\
UKMET&2
&$13.09\pm0.47$
&$19.41\pm1.12$
&$13.19\pm0.49$
&$19.37\pm1.10$
&$13.28\pm0.46$
&$19.43\pm1.12$
\\
UKMET&4
&$13.80\pm0.40$
&$22.72\pm0.98$
&$13.38\pm0.46$
&$22.71\pm0.97$
&$14.21\pm0.44$
&$22.77\pm0.98$
\\
UKMET&8
&$14.77\pm0.27$
&$28.15\pm0.70$
& N/A
&$28.24\pm0.74$

&$15.61\pm0.25$
&$28.35\pm0.73$
\\
AS&1
&$25.60\pm0.65$
&$66.32\pm18.85$
&$25.62\pm0.66$
&$66.32\pm18.85$
&$25.63\pm0.67$
&$66.32\pm18.85$
\\
AS&2
&$27.65\pm0.40$
&$111.69\pm19.91$
&$27.79\pm0.42$
&$111.69\pm19.91$
&$27.74\pm0.42$
&$111.69\pm19.91$
\\
AS&4
&$28.05\pm0.35$
&$136.81\pm14.18$
&$28.34\pm0.36$
&$136.81\pm14.18$
&$28.18\pm0.37$
&$136.81\pm14.18$
\\
AS&8
&$28.57\pm0.26$
&$139.46\pm13.04$
&$31.03\pm0.56$
&$139.46\pm13.04$
&$30.86\pm0.54$
&$139.46\pm13.04$
\\
HT&1
&$21.07\pm0.36$
&$23.13\pm0.49$
&$21.03\pm0.35$
&$23.13\pm0.49$
&$21.07\pm0.36$
&$23.13\pm0.49$
\\
HT&2
&$22.39\pm0.20$
&$40.79\pm16.67$
&$22.42\pm0.20$
&$40.79\pm16.67$
&$22.42\pm0.20$
&$40.79\pm16.67$
\\
HT&4
&$22.79\pm0.22$
&$33.81\pm9.95$
&$22.89\pm0.23$
&$33.81\pm9.95$
&$22.84\pm0.21$
&$33.81\pm9.95$
\\
HT&8
&$23.30\pm0.15$
&$34.51\pm6.93$
&$23.42\pm0.16$
&$34.51\pm6.93$
&$23.43\pm0.16$
&$34.51\pm6.93$
\\
    \bottomrule
    \end{tabular}}%
\end{table*}
\end{document}